\def\llncs{0}
\def\fullpage{1}
\def\anonymous{0}
\def\authnote{0}
\def\notxfont{0}
\def\submission{0}

\ifnum\submission=1
\def\llncs{1}
\fi

\ifnum\llncs=1
 \documentclass{llncs}
\else
	\documentclass[letterpaper,hmargin=1.05in,vmargin=1.05in,11pt]{article}
			\ifnum\fullpage=1
		\usepackage{fullpage}
		\fi
\fi

\usepackage[%
  colorlinks=true,
  citecolor=blue,
  pagebackref=true
]{hyperref}

\usepackage{amsmath, amsfonts, amssymb, mathtools,amscd}

\usepackage{amsthm}

\usepackage{lmodern}
\usepackage[T1]{fontenc}
\usepackage[utf8]{inputenc}

\usepackage{arydshln} 
\usepackage{url}
\usepackage{ifthen}
\usepackage{bm}
\usepackage{multirow}
\usepackage[dvips]{graphicx}
\usepackage[usenames]{color}
\usepackage{xcolor,colortbl} 
\usepackage{threeparttable}
\usepackage{comment}
\usepackage{paralist,verbatim}
\usepackage{cases}
\usepackage{booktabs}
\usepackage{braket}
\usepackage{cancel} 
\usepackage{ascmac} 
\usepackage{framed}
\usepackage{authblk}
\usepackage{pifont}
\usepackage{qcircuit}
\definecolor{darkblue}{rgb}{0,0,0.6}
\definecolor{darkgreen}{rgb}{0,0.5,0}
\definecolor{maroon}{rgb}{0.5,0.1,0.1}
\definecolor{dpurple}{rgb}{0.2,0,0.65}

\usepackage[capitalise,noabbrev]{cleveref}
\usepackage[absolute]{textpos}
\usepackage[final]{microtype}
\usepackage[absolute]{textpos}
\usepackage{everypage}
\DeclareMathAlphabet{\mathpzc}{OT1}{pzc}{m}{it}

\usepackage{algorithmic}
\usepackage{algorithm}
\usepackage{here}

\newtheoremstyle{thicktheorem}%
{\topsep}
{\topsep}
{\itshape}{}%
{\bfseries}%
{.}
{ }%
{\thmname{#1}\thmnumber{ #2}%
		\thmnote{ (#3)}%
}

\newtheoremstyle{remark}
{\topsep}
{\topsep}
	{}
	{}
	{}
	{.}
	{ }
	{\textit{\thmname{#1}}\thmnumber{ #2}
			\thmnote{ (#3)}%
	}

\ifnum\llncs=0
	\theoremstyle{thicktheorem}
	\newtheorem{theorem}{Theorem}[section]
	\newtheorem{lemma}[theorem]{Lemma}
	\newtheorem{corollary}[theorem]{Corollary}
	
	\newtheorem{definition}[theorem]{Definition}

	\theoremstyle{remark}
	
	\newtheorem{remark}[theorem]{Remark}

\else
\fi

\Crefname{MyClaim}{Claim}{Claims}

	\crefname{theorem}{Theorem}{Theorems}
	\crefname{assumption}{Assumption}{Assumptions}
	\crefname{construction}{Construction}{Constructions}
	\crefname{corollary}{Corollary}{Corollaries}
	\crefname{conjecture}{Conjecture}{Conjectures}
	\crefname{definition}{Definition}{Definitions}
	\crefname{exmaple}{Example}{Examples}
	\crefname{experiment}{Experiment}{Experiments}
	\crefname{counterexample}{Counterexample}{Counterexamples}
	\crefname{lemma}{Lemma}{Lemmata}
	\crefname{observation}{Observation}{Observations}
	\crefname{proposition}{Proposition}{Propositions}
	\crefname{remark}{Remark}{Remarks}
	\crefname{claim}{Claim}{Claims}
	\crefname{fact}{Fact}{Facts}
	\crefname{note}{Note}{Notes}

\ifnum\llncs=1
 \crefname{appendix}{App.}{Appendices}
 \crefname{section}{Sec.}{Sections}
\else
\fi

\ifnum\llncs=1
\renewcommand*{\backref}[1]{}
\else
	\renewcommand*{\backref}[1]{(Cited on page~#1.)}
	\ifnum\notxfont=1
	\else
		\usepackage{newtxtext}
	\fi
\fi
\ifnum\authnote=0  
\newcommand{\mor}[1]{}
\newcommand{\minki}[1]{}
\newcommand{\takashi}[1]{}

\else
\newcommand{\mor}[1]{$\ll$\textsf{\color{red} Tomoyuki: { #1}}$\gg$}
\newcommand{\takashi}[1]{$\ll$\textsf{\color{orange} Takashi: { #1}}$\gg$}
\newcommand{\minki}[1]{$\ll$\textsf{\color{darkgreen} Minki: { #1}}$\gg$}

\fi

\newcommand{\Tr}{\mathrm{Tr}}

\newcommand{\StateGen}{\mathsf{StateGen}}

\newcommand{\cA}{\mathcal{A}}
\newcommand{\cB}{\mathcal{B}}

\newcommand{\cN}{\mathcal{N}}

\def\makeuppercase#1{
\expandafter\newcommand\csname tl#1\endcsname{\widetilde{#1}}
}

\def\makelowercase#1{
\expandafter\newcommand\csname tl#1\endcsname{\widetilde{#1}}
}

\newcommand{\regC}{\mathbf{C}}
\newcommand{\regR}{\mathbf{R}}
\newcommand{\regZ}{\mathbf{Z}}
\newcommand{\regD}{\mathbf{D}}

\newcommand{\regB}{\mathbf{B}}
\newcommand{\regA}{\mathbf{A}}

\newcommand{\sep}{\lambda}
\newcommand{\secp}{\lambda}

\newcommand{\A}{\entity{A}}

\newcommand*{\algo}[1]{\ensuremath{\mathsf{#1}}}

\newcommand*{\entity}[1]{\mathcal{#1}}

\newenvironment{boxfig}[2]{\begin{figure}[#1]\fbox{\begin{minipage}{0.97\linewidth}
                        \vspace{0.2em}
                        \makebox[0.025\linewidth]{}
                        \begin{minipage}{0.95\linewidth}
            {{
                        #2 }}
                        \end{minipage}
                        \vspace{0.2em}
                        \end{minipage}}}{\end{figure}}

\newcommand{\bit}{\{0,1\}}

\newcommand{\KeyGen}{\algo{KeyGen}}

\newcommand{\Ver}{\algo{Ver}}

\newcommand{\tomo}{\mathsf{Tomography}}
\newcommand{\enet}{\mathsf{Net}}

\newcommand{\negl}{{\mathsf{negl}}}

\newcommand{\poly}{{\mathrm{poly}}}
\newcommand{\polylog}{{\mathrm{polylog}}}

\makeatletter
\DeclareRobustCommand
  \myvdots{\vbox{\baselineskip4\p@ \lineskiplimit\z@
    \hbox{.}\hbox{.}\hbox{.}}}
\makeatother

\title{A Note on Output Length of One-Way State Generators and EFIs}

\ifnum\anonymous=1
\ifnum\llncs=1
\author{\empty}\institute{\empty}
\else
\author{}
\fi
\else
\ifnum\llncs=1
\author{
 Minki Hhan\inst{1} \and Tomoyuki Morimae\inst{2} \and Takashi Yamakawa\inst{2,3,4}
}
\institute{KIAS, Seoul, Republic of Korea \and
 Yukawa Institute for Theoretical Physics, Kyoto University, Kyoto, Japan \and NTT Social Informatics Laboratories, Tokyo, Japan \and NTT Research Center for Theoretical Quantum Information, Atsugi, Japan
}
\else
\author[1]{Minki Hhan}
\author[2]{\hskip 1em Tomoyuki Morimae}
\author[2,3,4]{\hskip 1em Takashi Yamakawa}
\affil[1]{{\small KIAS, Seoul, Republic of Korea}\authorcr{\small minkihhan@kias.re.kr}}
\affil[2]{{\small Yukawa Institute for Theoretical Physics, Kyoto University, Kyoto, Japan}\authorcr{\small tomoyuki.morimae@yukawa.kyoto-u.ac.jp}}
\affil[3]{{\small NTT Social Informatics Laboratories, Tokyo, Japan}\authorcr{\small takashi.yamakawa.ga@hco.ntt.co.jp}}
\affil[4]{{\small NTT Research Center for Theoretical Quantum Information, Atsugi, Japan}}
\fi 
\fi

\date{}

\begin{document}

\maketitle

\begin{abstract}
We study the output length of one-way state generators (OWSGs), their weaker variants, and EFIs. 
\begin{itemize}
\item[\bf Standard OWSGs.] 
Recently, Cavalar et al. (arXiv:2312.08363) give OWSGs with $m$-qubit outputs for any $m=\omega(\log \lambda)$,
where $\lambda$ is the security parameter,
and conjecture that there do not exist OWSGs with $O(\log \log \lambda)$-qubit outputs. 
We prove their conjecture in a stronger manner by showing that there do not exist OWSGs with $O(\log \lambda)$-qubit outputs.
This means that their construction is optimal in terms of output length. 
\item[\bf Inverse-polynomial-advantage OWSGs.] 
Let $\epsilon$-OWSGs be a parameterized variant of OWSGs where a quantum polynomial-time adversary's advantage is at most $\epsilon$.  
For any constant $c\in \mathbb{N}$, 
we construct $\lambda^{-c}$-OWSGs with $((c+1)\log \lambda+O(1))$-qubit outputs  
assuming the existence of OWFs. 
We show that this is almost tight by proving that there do not exist $\lambda^{-c}$-OWSGs with at most $(c\log \lambda-2)$-qubit outputs. 
\item[\bf Constant-advantage OWSGs.]  
For any constant $\epsilon>0$, we construct $\epsilon$-OWSGs with $O(\log \log \lambda)$-qubit outputs  
assuming the existence of subexponentially secure OWFs. 
We show that this is almost tight by proving that there do not exist $O(1)$-OWSGs with $((\log \log \lambda)/2+O(1))$-qubit outputs. 
\item[\bf Weak OWSGs.]
We refer to  $(1-1/\poly(\lambda))$-OWSGs as weak OWSGs. 
We construct weak OWSGs with $m$-qubit outputs for any $m=\omega(1)$ assuming the existence of exponentially secure 
OWFs with linear expansion. 
We show that this is tight by proving that there do not exist weak OWSGs with $O(1)$-qubit outputs.

\item[\bf EFIs.]
We show that there do not exist $O(\log \lambda)$-qubit EFIs. We show that this is tight by proving that 
there exist $m(\lambda)$-qubit EFIs for any $m(\lambda)=\omega(\log \secp)$ assuming the existence of exponentially secure pseudorandom generators.
\end{itemize}
\end{abstract}

\section{Introduction}

In classical cryptography, the existence of one-way functions (OWFs) is the minimum assumption, because many primitives
(such as secret-key encryption (SKE), commitments, zero-knowledge, pseudorandom generators (PRGs), pseudorandom functions (PRFs), and digital signatures) are equivalent to OWFs, and
almost all cryptographic primitives imply OWFs~\cite{STOC:LubRac86,FOCS:ImpLub89,STOC:ImpLevLub89}. 
On the other hand, in quantum cryptography where quantum computing and quantum communications are available, 
recent studies suggest that OWFs would not be the minimum assumption,
and many new replacements of OWFs have been introduced.
For example, pseudorandom state generators (PRSGs)~\cite{C:JiLiuSon18}, one-way state generators (OWSGs)~\cite{C:MorYam22,TQC:MorYam24},
and EFIs~\cite{ITCS:BCQ23} are quantum analogue of PRGs, OWFs, and EFIDs~\cite{Gol90}, respectively. 
They seem to be weaker than OWFs~\cite{Kre21,STOC23:KreQiaSinTal,cryptoeprint:2023/1602},
yet can be used to construct several useful primitives such as private-key quantum money, SKE, commitments, multiparty computations, and digital signatures~\cite{C:MorYam22,C:AnaQiaYue22,ITCS:BCQ23,AC:Yan22}. 


OWSGs are a quantum analog of OWFs.\footnote{
There are a pure state output version~\cite{C:MorYam22} and mixed state output version~\cite{TQC:MorYam24} of OWSGs.  
We consider OWSGs with mixed state outputs as a default definition, but all of our positive results actually have pure state outputs. 
We adopt the mixed state version to make our negative results stronger.
} 
A OWSG is a quantum polynomial-time (QPT) algorithm
that takes a bit string $k$ as input and outputs a quantum state $\phi_k$. The security is, roughly speaking, that
no QPT adversary can find ``preimage'' of $\phi_k$.
There are two differences from classical OWFs. First, because quantum states cannot be copied in general, we have to take care of the number of copies of $\phi_k$ that
the adversary receives. In our definition, we assume that any polynomially many copies of $\phi_k$ are given to the adversary.
Second, we have to explicitly define what do we mean by ``preimage'' unlike classical OWFs, because two different quantum states can have some non-zero overlap.
We introduce a QPT verification algorithm that takes $k'$ (which is the output of the adversary) and $\phi_k$ as input, and
outputs $\top$ or $\bot$.\footnote{If $\phi_k$ is a pure state, we can assume without loss of generality that the verification algorithm is the following one:
project $\ket{\phi_k}$ onto $\ket{\phi_{k'}}$ and output $\top$ if the projection is successful. Otherwise, output $\bot$.}

Recently, Cavalar, Goldin, Gray, Hall, Liu, and Pelecanos~\cite{CGGHLP23}, among many other results, raised the question of how short output length of OWSGs can be. 
They show that PRSGs with $\omega(\log \secp)$-qubit outputs are also OWSGs,
where $\secp$ is the security parameter.\footnote{A PRSG is a QPT algorithm that takes a bit string as input and outputs a quantum state whose polynomially-many copies 
are computationally indistinguishable from the same number of copies of Haar random states.}
Since it is known that OWFs imply PRSGs of arbitrary (QPT-computable) output length~\cite{C:BraShm20}, their result gives OWSGs with $m(\secp)$-qubit outputs 
for any QPT-computable function $m(\secp)=\omega(\log \secp)$ based on OWFs.\footnote{
We say that a function $m$ defined on $\mathbb{N}$ is QPT-computable if there is a QPT algorithm that computes $m(\secp)$ on input $1^\secp$ with probability $1$. 
}  
Regarding a lower bound on the output length, they conjecture that there do not exist OWSGs with
$O(\log\log \secp)$-qubit outputs. However, they do not prove it and leave it open to study parameter regimes for which OWSGs do and do not exist. 

\takashi{I added the following paragraph.}
An EFI is a pair of
efficiently generatable states that are statistically far but computationally indistinguishable. 
While relationships between EFIs and other primitives have been actively studied \cite{TQC:MorYam24,ITCS:BCQ23}, to our knowledge, there is no existing work that studies how many qubits are needed for EFIs.       

\subsection{Our Results} 
\takashi{Changed the order of presenting the results}
\paragraph{\bf Lower Bounds for ($\epsilon$-)OWSGs.}
We prove that the conjecture of \cite{CGGHLP23} about the lower bound for the output length of OWSGs is true. 
Indeed, we prove a much stronger statement than their conjecture. 
\begin{theorem}\label{thm:intro_impossibility_log_OWSG}
There do not exist OWSGs with $O(\log\secp)$-qubit outputs.    
\end{theorem}
Note that their conjecture is the impossibility of OWSGs with $O(\log \log\secp)$-qubit outputs, but we prove the impossibility even for $O(\log\secp)$-qubit outputs. 
The proof of the above theorem is quite simple. We simply consider a trivial attacker that just outputs a random input $k'$ just ignoring the target state $\phi_k$. 
When $\phi_k$ is an $m$-qubit \emph{pure} state, Welch bound~\cite{Welch}~(\Cref{lem:Welch}) directly implies that the above adversary has an advantage at least $2^{-m}$, which is inverse-polynomial when $m=O(\log\secp)$. 
For the case of mixed state outputs, we need slightly more calculations, but we can still show that the above adversary has an inverse-polynomial advantage when $m=O(\log\secp)$.  

The above theorem means that when the output length is $O(\log\secp)$, there is an adversary with advantage $1/\poly(\secp)$. 
On the other hand,  for OWFs with $O(\log\secp)$-bit outputs, there is an adversary with advantage $1-1/\poly(\secp)$ for any polynomial $\poly$ (e.g., see \cite[Excercise 10 of Section 2]{DBLP:books/cu/Goldreich2001}). 
Thus, it is natural to ask if such an adversary with advantage $1-1/\poly(\secp)$ exists for OWSGs with $O(\log\secp)$-qubit outputs. 
First, we note that we cannot amplify the adversary's advantage by repeating it many times in general.\footnote{Note that such an amplification of the adversary's advantage does not work even for OWFs. Indeed, if that was possible, then any weak OWF would be also a OWF. But we can easily construct a weak OWF that is not a OWF assuming the existence of OWF. 
Do not confuse this with the amplification theorem that constructs OWFs from weak OWFs~\cite{FOCS:Yao82a}.
The amplification theorem means that we can convert any weak OWF to a OWF, but it does not mean that any weak OWF is also a OWF.
}  
For example, suppose that an adversary is given many copies of $\phi_k$
always outputs $k'$ such that the verification algorithm accepts $(k',\phi_{k})$ with probability $1/\poly(\secp)$. 
Such an adversary has an advantage $1/\poly(\secp)$ against the security of the OWSG, but one cannot gain anything by running it many times since it always outputs $k'$ that has only $1/\poly(\secp)$ acceptance probability.
This leaves a possibility that we can circumvent the barrier of $O(\log\secp)$-qubit outputs if we weaken the security requirement. To better understand the tradeoffs between the output length and security, we introduce an $\epsilon$-OWSGs, which are a parameterized variant of OWSGs where the adversary's advantage is only required to be at most $\epsilon$, and study parameter regimes for output length for which $\epsilon$-OWSGs do and do not exist. 

First, we study the regime of $\epsilon=\secp^{-\Theta(1)}$. We show that the proof of \Cref{thm:intro_impossibility_log_OWSG} naturally extends to this setting to obtain the following theorem. 
\begin{theorem}\label{thm:intro_impossibility_inverse-poly_OWSG}  
For any constant $c\in \mathbb{N}$,   
there do not exist $\secp^{-c}$-OWSGs of which output length is at most $c \log\secp-2$.      
\end{theorem}

Next, we study the regime of $\epsilon=O(1)$. 
We show the following lower bound.  
\begin{theorem}\label{thm:intro_impossibility_loglog_constant_OWSG}
There do not exist $O(1)$-OWSGs
with $(\frac{1}{2}\log\log \secp+O(1))$-qubit outputs. 
\end{theorem}
In this regime, the idea for the proofs of  \Cref{thm:intro_impossibility_log_OWSG,thm:intro_impossibility_inverse-poly_OWSG} does not work since Welch bound only gives an adversary with advantage $1/\poly(\secp)$. Thus, we rely on a different idea. 
Our proof is inspired by the construction of quantum pseudorandom generators from short PRSGs~\cite{cryptoeprint:2023/904}. If their construction applies to the one-wayness, we can construct weak OWFs with $O(\log \secp)$-bit output length, which is known to be impossible. Unfortunately, their result heavily relies on the pseudorandomness of outputs, and extending it to the OWSGs does not seem easy. Instead, we use the tomography and epsilon-nets\footnote{An $\gamma$-net of $m$-qubit quantum states is a set of quantum states such that for any $m$-qubit quantum state, there exists an $\gamma$-close element in the set. We use $\gamma$ instead of $\epsilon$ to avoid confusion with the security of OWSGs.} (\Cref{lem:tomography,lem:epsnet}) to construct a randomized algorithm with classical outputs from the OWSGs, in which we can adapt the impossibility of weak OWFs with $O(\log \secp)$-bit output length to break the original OWSGs. The output length is bounded by $O(\log \lambda)$ because the cardinality of $\gamma$-net in $m$-qubit space is bounded by $(C/\gamma)^{2^{2m}}$ for some constant $C$, which is polynomial for constant $\gamma$ and 
$m=\frac{1}{2}\log\log \secp+O(1)$. 
For the OWSGs with pure state outputs, this can be improved to $\log\log \secp +O(1).$

Finally, we study the regime of  $\epsilon=1-1/\poly(\secp)$ for some polynomial $\poly$. 
We refer to $(1-1/\poly(\secp))$-OWSGs as weak OWSGs following \cite{TQC:MorYam24}. 
We show the following lower bound. 
\begin{theorem}\label{thm:intro_impossibility_constant_weak_OWSG}
There do not exist weak OWSGs with $O(1)$-qubit outputs. 
\end{theorem} 
The proof is almost the same as that of \Cref{thm:intro_impossibility_loglog_constant_OWSG} 
noting that $(C/\gamma)^{2^{2m}}$ is polynomial when $\gamma$ is inverse-polynomial and $m$ is constant.

\paragraph{\bf Upper Bounds for ($\epsilon$)-OWSGs.}
We show the above lower bounds are almost tight by showing almost matching upper bounds. 

As already mentioned, the combination of existing works \cite{C:BraShm20,CGGHLP23} gives 
$m(\secp)$-qubit output OWSGs based on OWFs for any QPT-computable function $m(\secp)=\omega(\log \secp)$.
This means that \Cref{thm:impossibility_log_OWSG} is tight. 
We show that their construction can be extended to $\secp^{-\Theta(1)}$-OWSGs to obtain the following theorem. 
\begin{theorem}\label{thm:intro_possibility_inverse-poly_OWSG}
   For any constant $c\in \mathbb{N}$, there exist $\lambda^{-c}$-OWSGs with $((c+1) \log \lambda+O(1))$-qubit pure state outputs assuming the existence of OWFs.
\end{theorem}  
Comparing this with \Cref{thm:intro_impossibility_inverse-poly_OWSG}, we can see that it is almost tight though there is a small gap between $c\log \secp-2$ and $(c+1)\log \secp+O(1)$. 
The above theorem means that we can circumvent the barrier of $O(\log\secp)$-qubit outputs for OWSGs (\Cref{thm:intro_impossibility_log_OWSG}) if we slightly weaken ($\negl(\secp)$-)OWSGs to $\secp^{-\Theta(1)}$-OWSGs.

Next, we show an upper bound for the output length of $\Theta(1)$-OWSGs.  
\begin{theorem}\label{thm:intro_possibility_Ologlog_constant_OWSG}
For any constant $\epsilon>0$, there exist $\epsilon$-OWSGs with $O(\log\log \secp)$-qubit pure state outputs assuming the existence of subexponentially secure OWFs.   
\end{theorem} 
Comparing this with \Cref{thm:intro_impossibility_loglog_constant_OWSG}, we can see that it is tight ignoring the constant factor for $\log\log \secp$. 
We show this theorem by using quantum fingerprinting~\cite{fingerprinting}. Roughly, quantum finger printing encodes an $\ell$-bit classical string $x$ into an $O(\log(\ell)\log(\eta^{-1}))$-qubit pure state $\ket{h_x}$ in such a way that $|\langle h_x | h_{x'}\rangle|\le \eta$ for any $x\neq x'$. 
Our idea is to compose a OWF and quantum fingerprinting, i.e., our OWSG outputs $\ket{h_{f(x)}}$ on input $x$ where $f$ is a OWF.  
When $\eta=\Theta(1)$, the length of the output is logarithmic in the output length of $f$. Thus, if we start from a OWF with $\polylog(\secp)$-bit outputs, which exists assuming the existence of subexponentially secure OWFs, the output length of the above OWSG is $O(\log\log \secp)$. 
We also show that the above construction preserves the one-wayness up to an additive constant error caused by quantum fingerprinting. This yields \Cref{thm:intro_possibility_Ologlog_constant_OWSG}.


Finally, we show an upper bound for the output length of weak OWSGs. 
\begin{theorem}\label{thm:intro_possibility_superconstant_weak_OWSG}
For any QPT-computable 
function $m(\secp)=\omega(1)$, there exist weak OWSGs with $m(\secp)$-qubit pure state outputs assuming the existence of exponentially secure OWFs that map $\secp$-bit inputs to $O(\secp)$-bit outputs. 
\end{theorem} 
Comparing this with \Cref{thm:intro_impossibility_constant_weak_OWSG}, we can see that this is tight. 
The main observation for proving this theorem is the following. 
For an angle $\theta$, let $\ket{+_\theta}:=\frac{1}{\sqrt{2}}\left(\ket{0}+e^{i\theta}\ket{1}\right)$. Then we can show that for any distinct $y,y'\in [\secp]$, we have $|\langle +_{2\pi y/\secp}|+_{2\pi y'/\secp}\rangle|\le \sqrt{\frac{1+\cos (2\pi/\secp)}{2}}=1-1/\poly(\secp)$. 
This means that we can encode $\log \secp$ classical bits into $1$ qubit if we allow the inner products to be as large as $1-1/\poly(\secp)$. 
By repeating it, we can encode  $t\log \secp$ classical bits into $t$ qubits. 
By using this encoding instead of quantum fingerprinting in the construction for \Cref{thm:intro_possibility_Ologlog_constant_OWSG}, we obtain weak OWSGs with $t$-qubit outputs starting from any OWFs with $t\log \secp$-bit outputs. 
Assuming the existence of exponentially secure OWFs with linear expansion, there are OWFs with $t\log \secp$-bit outputs for any super-constant $t$. 
This yields \Cref{thm:intro_possibility_superconstant_weak_OWSG}.


Our results 
on upper and lower bounds for ($\epsilon$-)OWSGs 
are summarized in \Cref{table:results}.

\begin{table}[h]
\caption{Summary of the results}
 \label{table:results}
 \centering
  \begin{tabular}{cll}
   \hline
   Security of OWSGs& Upperbound of output length& Lowerbound of output length\\
   \hline \hline
   $\negl(\secp)$  & $\omega(\log\secp)~\cite{C:BraShm20,CGGHLP23}$
& $O(\log\secp)$ [\cref{thm:intro_impossibility_log_OWSG}] \\
&(Assuming OWFs)&\\
\hline
$\secp^{-c}$  & $(c+1)\log\secp+O(1)$ [\cref{thm:intro_possibility_inverse-poly_OWSG}]
& $c\log\secp-2$ [\cref{thm:intro_impossibility_inverse-poly_OWSG}] \\
$(c=\Theta(1))$&(Assuming OWFs)&\\
\hline
constant& $O(\log\log\secp)$ [\cref{thm:intro_possibility_Ologlog_constant_OWSG}] & $\frac{1}{2}\log\log\secp+O(1)$ [\cref{thm:intro_impossibility_loglog_constant_OWSG}]\\
&(Assuming subexp secure OWFs)&\\
\hline
$1-\frac{1}{\poly(\secp)}$& $\omega(1)$ [\cref{thm:intro_possibility_superconstant_weak_OWSG}] & $O(1)$[\cref{thm:intro_impossibility_constant_weak_OWSG}] \\
& (Assuming exp secure OWFs) & \\
\hline
  \end{tabular}
\end{table}

\takashi{Added the following two paragraphs.}
\paragraph{\bf EFIs.} 
We show matching lower and upper bounds for the length of EFIs.
The lower bound is stated as follows. 
\begin{theorem}\label{thm:intro_lower_EFI}
There do not exist $O(\log\secp)$-qubit EFIs.
\end{theorem}
We remark that the proof of this theorem is straightforward in the non-uniform setting: Since the two states that form an EFI are statistically far, there exists a measurement that distinguishes them with advantage $1/\poly(\secp)$. If they are $O(\log \secp)$-qubit states, the optimal distinguishing measurement can be described as a $\poly(\secp)$-bit classical string. Thus, a non-uniform adversary that takes the description of the measurement as advice can break the computational indistinguishability. To extend it to the uniform setting, our idea is to use tomography.\footnote{We thank Fermi Ma for suggesting the use of tomography.} 
Roughly, we show that a uniform adversary can use tomography on the two states to efficiently approximate the optimal distinguishing measurement. 

We show a matching upper bound as follows.
\begin{theorem}\label{thm:intro_upper_EFI}
If exponentially secure PRGs exist, then
$m(\secp)$-qubit 
EFIs exist for any QPT computable function $m(\secp)=\omega(\log\secp)$.
\end{theorem}
The proof is almost obvious. By using an exponentially secure PRG $G:\bit^{m/2}\rightarrow \bit^m$, let $\rho_0:=\frac{1}{2^{m/2}}\sum_{x\in \bit^{m/2}}\ket{G(x)}\bra{G(x)}$ 
and $\rho_1:=\frac{1}{2^m}\sum_{y\in \bit^{m}}\ket{y}\bra{y}$. Then $\rho_0$ and $\rho_1$ are computationally indistinguishable by exponential security of the PRG and $m(\secp)=\omega(\log\secp)$, but they are statistically far since the image size of $G$ is at most $2^{m/2}\ll2^m$.  

\paragraph{\bf Quantum bit commitments.}
It is known that EFIs exist if and only if quantum bit commitments exist~\cite{AC:Yan22,ITCS:BCQ23}.  
However, the conversion between them does not preserve the number of qubits, and thus the lower and upper bounds for the length of EFIs do not necessarily imply similar bounds for quantum bit commitments. To complement it, we also show similar lower and upper bounds for the commitment length (i.e., the number of qubits of the commitment register) of quantum bit commitments. 
That is, we show a $O(\log \secp)$ lower bound and $\omega(\log \secp)$ upper bound assuming the existence of exponentially secure PRGs.  
See \Cref{sec:commitment} for the details.

\subsection{Discussion}\label{sec:discussion}
Our results suggest that the upper and lower bounds for the output length of ($\epsilon$-)OWSGs behave quite differently from those of OWFs.
\begin{itemize}
\item 
It is easy to see that weak OWFs with $O(\log \secp)$-bit outputs do not exist.  
On the other hand, assuming the existence of exponentially secure OWFs with linear expansion, we can show that there exist (strong) OWFs with $m$-bit outputs for any $m=\omega(\log \secp)$ by a simple padding argument. 
This means that $O(\log \secp)$ is the best lower bound for both OWFs and weak OWFs. 
On the other hand, our results suggest that the lower bounds for OWSGs and weak OWSGs are different. 
In particular, we show that there do not exist OWSGs with $O(\log \secp)$-qubit outputs (\Cref{thm:intro_impossibility_log_OWSG}), but there exist weak OWSGs with slightly super-constant-qubit outputs (\Cref{thm:intro_possibility_superconstant_weak_OWSG}), which circumvents the length barrier for standard OWSGs (\Cref{thm:intro_impossibility_log_OWSG}). \mor{OWFs?}\takashi{I think OWSGs is correct} \minki{In that case, I guess we need to explain which barrier it circumvents. You probably mean ``circumvents the length barrier for the standard OWSGs.''?}\takashi{clarified it}
\item 
For constructing OWFs with the optimal output length, i.e., those with slightly super-logarithmic-bit outputs, we need to assume exponential hardness of $\mathbf{NP}$ due to an obvious reason. On the other hand, we can construct OWSGs with slightly super-logarithmic-bit outputs only assuming polynomial hardness of OWFs~\cite{C:BraShm20,CGGHLP23}. 
Since we show that their construction is optimal in terms of output length (\Cref{thm:intro_impossibility_log_OWSG}), this means that polynomial hardness is sufficient for achieving the optimal output length for OWSGs. 
\end{itemize}

Our upper and lower bounds for the output length of OWSGs also imply those for
several other primitives.
\begin{enumerate}
    \item 
    {\bf Quantum-public-key digital signatures:}
    OWSGs are equivalent to (bounded poly-time secure) digital signatures with quantum public keys~\cite{TQC:MorYam24}.
    Hence we have the following consequences.
    First, from \cref{thm:intro_impossibility_log_OWSG},
    the length of the quantum public key of digital signatures cannot be $O(\log\secp)$.
    Second, 
    from \cref{thm:OWSG_with_slightly_poly},
    assuming OWFs, there exist one-time secure digital signatures with $m$-qubit quantum public keys for any $m=\omega(\log(\secp))$. This can be regarded as a quantum advantage because it is unlikely that we can construct classical one-time secure
    digital signatures with similar public key length from (polynomially secure) OWFs. 
    \item 
    {\bf Private-key quantum money:}
    Private-key quantum money schemes with pure banknotes imply OWSGs~\cite{TQC:MorYam24}.
    Hence, from \cref{thm:intro_impossibility_log_OWSG},
    the length of pure quantum banknotes cannot be $O(\log\secp)$.
\end{enumerate}

\paragraph{\bf Open problems.} 
Though we give matching upper and lower bounds for the output length of OWSGs and weak OWSGs, there are gaps for 
$\secp^{\Theta(1)}$-OWSGs and 
$\Theta(1)$-OWSGs.
For $\secp^{c}$-OWSGs with $c=\Theta(1)$, our lower bound is $c \log \secp -2$ (\Cref{thm:intro_impossibility_inverse-poly_OWSG}) whereas our upper bound is $(c+1)\log \secp+O(1)$ (\Cref{thm:intro_possibility_inverse-poly_OWSG}).  
For $\Theta(1)$-OWSGs, our lower bound is $(\log\log \secp)/2+O(1)$ (\Cref{thm:intro_impossibility_loglog_constant_OWSG}) whereas upper bound is $O(\log\log \secp)$ (\Cref{thm:intro_possibility_Ologlog_constant_OWSG}). 
It is an interesting open problem to fill the gap. 

Another interesting question is to construct quantum OWFs using short output OWSGs.
Here, a quantum OWF is a quantum (pseudo)deterministic polynomial-time algorithm whose inputs and outputs are classical bit strings and the standard one-wayness is satisfied.
The construction of quantum PRGs from short output PRSGs in~\cite{cryptoeprint:2023/904} would not be applied to the case of OWSGs,
because both the proofs for (pseudodeterministic) correctness and security exploit (pseudo)randomness of output states of PRSGs. Answering this question provides a better understanding of quantum cryptographic primitives with short output lengths, and simplifies the proofs of~\Cref{thm:intro_impossibility_loglog_constant_OWSG,thm:intro_impossibility_constant_weak_OWSG}.

\subsection{Organization}
We prepare definitions and technical tools in \Cref{sec:preliminaries}. 
We prove our results on lower bounds for the output length of ($\epsilon$-)OWSGs (\Cref{thm:intro_impossibility_log_OWSG,thm:intro_impossibility_inverse-poly_OWSG,thm:intro_impossibility_constant_weak_OWSG,thm:intro_impossibility_loglog_constant_OWSG}) in \Cref{sec:lower_bound}.
We prove our results on upper bounds for the output length of ($\epsilon$-)OWSGs (\Cref{thm:intro_possibility_Ologlog_constant_OWSG,thm:intro_possibility_inverse-poly_OWSG,thm:intro_possibility_superconstant_weak_OWSG}) in \Cref{sec:upper_bound}.
We show our results on lower and upper bounds for the length of EFIs (\Cref{thm:intro_lower_EFI,thm:intro_upper_EFI}) in \Cref{sec:EFI}. 
\section{Preliminaries}\label{sec:preliminaries}
\paragraph{\bf Notations.} 
We use standard notations of quantum computing and cryptography.
We use $\secp$ as the security parameter.
$[n]$ means the set $\{1,2,...,n\}$.
For a finite set $S$, $x\gets S$ means that an element $x$ is sampled uniformly at random from the set $S$.
$\negl$ is a negligible function, and $\poly$ is a polynomial.
PPT stands for (classical) probabilistic polynomial-time and QPT stands for quantum polynomial-time. 
We say that a function $m$ defined on $\mathbb{N}$ is QPT-computable if there is a QPT algorithm that computes $m(\secp)$ on input $1^\secp$ with probability $1$. We stress that the running time of the algorithm can be polynomial in $\secp$ rather than in $\log \secp$. 
For an algorithm $A$, $y\gets A(x)$ means that the algorithm $A$ outputs $y$ on input $x$.
For quantum states $\rho,\sigma$, $\|\rho-\sigma\|_{tr}$ means their trace distance. 
For any pure states $\ket{\psi},\ket{\phi}$ it is known that $\||\psi\rangle\langle\psi|-|\phi\rangle\langle\phi|\|_{tr}=\sqrt{1-|\langle\psi |\phi \rangle|^2}$. 

\subsection{Basic Cryptographic Primitives}
\begin{definition}[OWFs]
\label{def:OWFs}
A function $f:\bit^*\to\bit^*$ is a (quantum-secure) OWF if
it is computable in classical deterministic polynomial-time, and
for any QPT adversary $\cA$, 
\begin{equation}
\Pr[f(x')=f(x):
x\gets\bit^\secp,
x'\gets\cA(1^\secp,f(x))
]
\le\negl(\secp).
\end{equation} 

We say that $f$ is a subexponentially secure OWF if there is a constant $0<c < 1$ such that
for any QPT adversary $\cA$, 
\begin{equation}
\Pr[f(x')=f(x):
x\gets\bit^\secp,
x'\gets\cA(1^\secp,f(x))
]
\le 2^{-\secp^{c}}
\end{equation}
for all sufficiently large $\secp$. 

We say that $f$ is an exponentially secure OWF if there is a constant $0<c < 1$ such that
for any QPT adversary $\cA$, 
\begin{equation}
\Pr[f(x')=f(x):
x\gets\bit^\secp,
x'\gets\cA(1^\secp,f(x))
]
\le 2^{-c \secp}
\end{equation}
for all sufficiently large $\secp$.  
\end{definition}
\begin{remark}
We require OWFs to be computable in classical deterministic polynomial-time following the standard definition. However, all of our positive results still follow even if we relax the requirement to only require them to be computable in QPT. 
\end{remark}

\begin{definition}[PRGs] \takashi{I added this definition}
A function $G:\bit^*\to\bit^*$ is a (quantum-secure) PRG if
it is computable in classical deterministic polynomial-time, 
it maps $\secp$-bit inputs to $\ell(\secp)$-bit outputs for $\ell(\secp)\ge \secp+1$ and
for any QPT adversary $\cA$, 
\begin{align}
\left|\Pr_{x\gets\bit^\secp}[1\gets\cA(1^\secp,G(x))]    
-\Pr_{y\gets\bit^{\ell}}[1\gets\cA(1^\secp,y)]    \right|\le \negl(\secp). 
\end{align}
We say that a PRG is exponentially secure if there exists a constant $0<c<1$ such that for any QPT adversary $\cA$, 
\begin{align}
\left|\Pr_{x\gets\bit^\secp}[1\gets\cA(1^\secp,G(x))]    
-\Pr_{y\gets\bit^{\ell}}[1\gets\cA(1^\secp,y)]    \right|\le 2^{-c\secp}.
\end{align}
\end{definition}

\begin{definition}[One-way state generators (OWSGs)~
\cite{TQC:MorYam24}]
\label{def:OWSG}
For a function $\epsilon:\mathbb{N}\rightarrow [0,1]$, 
an $\epsilon$-one-way state generator ($\epsilon$-OWSG) is a tuple of QPT algorithms
$(\KeyGen,\StateGen,\Ver)$ that work as follows: 
\begin{itemize}
\item
$\KeyGen(1^\secp)\to k:$
It is a QPT algorithm that, on input the security parameter $\secp$, outputs a classical key $k$. 
    \item 
    $\StateGen(k)\to \phi_k:$ It is a QPT algorithm that, on input $k$, outputs 
    a quantum state $\phi_k$. 
    \item
    $\Ver(k',\phi_k)\to\top/\bot:$ It is a QPT algorithm that, on input $\phi_k$ and a bit string $k'$, outputs $\top$ or $\bot$. 
\end{itemize}

We require the following correctness and $\epsilon$-one-wayness.

\paragraph{\bf Correctness:}
\begin{eqnarray*}
\Pr[\top\gets\Ver(k,\phi_k):k\gets\KeyGen(1^\secp),\phi_k\gets\StateGen(k)]\ge1-\negl(\secp).
\end{eqnarray*}

\paragraph{\bf $\epsilon$-one-wayness:}
For any QPT adversary $\cA$ and any polynomial $t$\footnote{$\StateGen$ is actually run $t$ times to generate $t$ copies of $\phi_k$, but for simplicity, we just write
$\phi_k\gets\StateGen(k)$ only once. This simplification will often be used in this paper.},
\begin{eqnarray*}
\Pr[\top\gets\Ver(k',\phi_k):k\gets\KeyGen(1^\secp),\phi_k\gets\StateGen(k),k'\gets\cA(1^\secp,\phi_k^{\otimes t})]\le\epsilon (\secp)
\end{eqnarray*}
for all sufficiently large $\secp$. 

We say that it is a OWSG if it is a $\secp^{-c}$-OWSG for all $c\in \mathbb{N}$.

We say that it is a weak OWSG if there is $c\in \mathbb{N}$ such that it is a $(1-\secp^{-c})$-OWSG. 

We say that it has $m(\secp)$-qubit outputs if $\phi_k$ is an $m(\secp)$-qubit state for all $k$ in the support of $\KeyGen(1^\secp)$. 

We say that it has pure output states if $\phi_k$ is a pure state for all $k$ in the support of $\KeyGen(1^\secp)$. 
\end{definition}

\begin{definition}[Pseudo-random quantum states generators (PRSGs)~\cite{C:JiLiuSon18}]
    For a function $\epsilon:\mathbb N \to [0,1]$,
    An $\epsilon$-pseudorandom quantum states generator ($\epsilon$-PRSG) is a tuple of QPT algorithms $(\KeyGen,\StateGen)$ that work as follows:
    \begin{itemize}
        \item $\KeyGen(1^\lambda)\to k:$ It is a QPT algorithm that, on input the security parameter $\lambda$, 
        outputs a classical key $k$.
        \item $\StateGen(k)\to \ket{\phi_k}:$ It is a QPT algorithm that, on input $k$, outputs an $m$-qubit quantum state $\ket{\phi_k}$.
    \end{itemize}
    We require the following pseudorandomness: For any polynomial $t$ and any QPT algorithm $\cA$,
    \[
    \left| 
        \Pr_{k \gets \KeyGen(1^\lambda)} \left[ 
            1 \gets \cA\left(1^\lambda,
                \ket{\phi_k}^{\otimes t}
            \right)
        \right]
        -
        \Pr_{\ket{\psi}\gets \mu_m} \left[ 
            1 \gets \cA\left(1^\lambda,
                \ket{\psi}^{\otimes t}
            \right)
        \right]
    \right|
    \le \epsilon(\lambda)
    \]
    for all sufficiently large $\secp$, where $\mu_m$ is the Haar measure on $m$-qubit quantum states.
    We say that it is a PRSG if it is a $\lambda^{-c}$-PRSG for all $c\in \mathbb N$.
\end{definition}

\begin{definition}[EFIs~\cite{ITCS:BCQ23}]
\label{def:EFI} 
\takashi{added this definition}
An EFI is an algorithm $\StateGen(b,1^\secp)\to\rho_b$ that, on input $b\in\bit$ and the security parameter $\lambda$,
outputs a quantum state $\rho_b$ such that all of the following three conditions are satisfied.
\begin{itemize}
\item
It is a uniform QPT algorithm.
\item
$\rho_0$ and $\rho_1$ are computationally indistinguishable. In other words,
for any QPT adversary $\cA$,
$
|\Pr[1\gets\cA(1^\secp,\rho_0)]
-\Pr[1\gets\cA(1^\secp,\rho_1)]|\le\negl(\secp).
$
\item
$\rho_0$ and $\rho_1$ are statistically distinguishable,
i.e.,
$
\|\rho_0-\rho_1\|_{tr}\ge\frac{1}{\poly(\secp)}.
$
\end{itemize}
We call it an $\ell(\secp)$-qubit EFI if $\rho_0$ and $\rho_1$ are  $\ell(\secp)$-qubit states. 
\end{definition}
\begin{remark}
In the original definition of EFIs in \cite{ITCS:BCQ23}, computational indistinguishability is required to hold against non-uniform QPT adversaries whereas we only require computational indistinguishability against uniform QPT adversaries. 
This only makes our result on the lower bound for the output length of EFIs stronger.  
\end{remark}

\subsection{Useful Lemmata}
\begin{lemma}[Welch Bound~\cite{Welch,BelovsS08}]\label{lem:Welch}
Let $n,d,k$ be positive integers. 
For any family $\{\ket{\psi_i}\}_{i\in [n]}$ of (not necessarily unit) vectors in $\mathbb{C}^d$, it holds that 
\begin{align*}
{d+k-1 \choose k}\sum_{i,j \in [n]}|\langle \psi_i | \psi_j\rangle |^{2k}\ge  \left(\sum_{i\in [n]}\langle \psi_i | \psi_i\rangle^k \right)^2.
\end{align*}
\end{lemma}
In particular, we use the following corollary. 
\begin{corollary}\label{cor:Welch}
Let $n,m$ be positive integers. 
For any family $\{\ket{\phi_i}\}_{i\in [n]}$ of (normalized) $m$-qubit pure states and any distribution $\mathcal{D}$ over $[n]$, it holds that 
\begin{align*}
\mathbb{E}_{i\gets \mathcal{D},j\gets \mathcal{D}}|\langle \phi_i | \phi_j\rangle |^2\ge 2^{-m}.
\end{align*}
\end{corollary}
\begin{proof}
For each $i\in [n]$, 
let $p_i:=\Pr[i\gets \mathcal{D}]$ and
$\ket{\psi_i}:=\sqrt{p_i}\ket{\phi_i}$.
Applying \Cref{lem:Welch} for $\{\ket{\psi_i}\}_{i\in [n]}$
  where $k:=1$ and
  $d:=2^m$, we obtain
  \begin{align*}
      2^m \sum_{i,j \in [n]}p_i p_j|\langle \phi_i | \phi_j\rangle |^2 \ge \left(\sum_{i\in [n]}p_i\langle \phi_i | \phi_i\rangle \right)^2.
  \end{align*}
By definition, we have 
\begin{align*}
   \sum_{i,j \in [n]}p_i p_j|\langle \phi_i | \phi_j\rangle |^2=  \mathbb{E}_{i\gets \mathcal{D},j\gets \mathcal{D}}|\langle \phi_i | \phi_j\rangle |^2.
\end{align*}
On the other hand, since $\ket{\phi_i}$ is a normalized state, 
we have 
\begin{align*}
    \left(\sum_{i\in [n]}p_i\langle \phi_i | \phi_i\rangle \right)^2= \left(\sum_{i\in [n]}p_i\right)^2=1^2=1.
\end{align*}
Combining the above, we obtain \Cref{cor:Welch}. 
\end{proof}



\begin{lemma}[Quantum state tomography~{\cite[Corollary 2.16]{cryptoeprint:2023/904}}]
\label{lem:tomography} 
    Let $\delta=\delta(\lambda) \in (0,1]$ and $d=d(\lambda) \in \mathbb{N}$.
    There exists a tomography procedure $\tomo$ that takes $t(\lambda):=144\lambda d^4 /\delta^2$ copies of a $d$-dimensional mixed quantum state $\rho$, and outputs a Hermitian 
    matrix $M \in \mathbb{C}^{d \times d}$ such that
    \[
        \Pr \left[
            \|M-\rho \|_{tr} \le \delta : M \gets \tomo(\rho^{\otimes t},\delta)
        \right]
        \ge 
        1 -\negl(\lambda).
    \]
    Moreover, the running time of $\tomo$ is polynomial in $1/\delta,d$ and $\lambda$.
\end{lemma}
\begin{remark}
    \cite[Corollary 2.16]{cryptoeprint:2023/904} does not mention that the output matrix is Hermitian. This fact can be checked by following the proof; the starting point is the folklore tomography (see~\cite[Section 1.5.3]{Lowe21}) that approximates $\rho = \frac 1d \sum \Tr(P_i \rho) P_i$ for Pauli matrices $P_i$ by $N = \sum \mu_i P_i$ for $\mu_i \approx  \frac 1d \Tr(P_i \rho)$. The matrix $N$ is Hermitian, and the above lemma, which is a trace-distance version of~\cite[Corollary 7.6]{TCC:Luowen}, only chooses a nice $M$ among many independent $N$'s, so that $M$ is still Hermitian.
\end{remark}


\begin{lemma}[$\gamma$-net of quantum states~{\cite[Section 5]{ABMB}}]
\label{lem:epsnet}
    Let $\gamma\in (0,1]$ and $d \in \mathbb{N}$. There exist some constant $C>0$ and a set $\enet(d,\gamma)$ of density matrices, called an \emph{$\gamma$-net} (of quantum states), of size $|\enet(d,\gamma)| \le (C/\gamma)^{d^2}$ such that the following holds: For any $d$-dimensional mixed quantum state $\rho$, there exists $N \in \enet(d,\gamma)$ such that $\|N-\rho\|_{tr} \le \gamma$. Moreover, there exists an explicit construction of $\enet(d,\gamma)$ in time $|\enet(d,\gamma)|\cdot \poly(d,1/\gamma)$. If we only consider the pure states, $|\enet(d,\gamma)| \le (C/\gamma)^{d}$ holds. 
\end{lemma}
\begin{proof}[Sketch of proof]
    By purification, it suffices to show that there exists an explicit $\gamma$-net for pure states in dimension $d^2$ of size $(C/\gamma)^{d^2}$. This is given in~\cite[Exercise 5.22]{ABMB} by normalizing an (explicit) $\gamma$-net in the Boolean cube for Hamming distance in the same dimension.
\end{proof}

\begin{lemma}[Quantum Fingerprinting~\cite{fingerprinting}] 
\label{lem:finger}
There is a $\poly(\ell)\log\eta^{-1}$-time quantum algorithm that takes $x\in\bit^\ell$ as input 
and
outputs an $O(\log(\ell)\log(\eta^{-1}))$-qubit state $\ket{h_x}$ such that
$|\langle h_x|h_{x'}\rangle|\le\eta$ for each $x\neq x'$. 
\end{lemma}

\begin{proof}
Because the above statement is not explicitly written in \cite{fingerprinting}, we here provide its proof.
Let $E:\bit^\ell\to\bit^m$ be an error-correcting code
such that $m=c\ell$ and the Hamming distance between any two
distinct code words $E(x)$ and $E(y)$ 
(i.e., the number
of bit positions where they differ)
is at least $(1-\delta)m$, where $c$ and $\delta$ are positive constants.
For some specific codes, we can choose any $c>2$ and have
$\delta<\frac{9}{10}+\frac{1}{15c}$ for sufficiently large $\ell$~\cite{fingerprinting}.
Let us take $c=4$. Then $\delta<\frac{9}{10}+\frac{1}{15c}=\frac{11}{12}$ for sufficiently large $\ell$.
Let $r\coloneqq\lfloor \frac{\log\eta}{\log\frac{11}{12}}\rfloor+1$.
For each $x\in\bit^\ell$, define
the state
\begin{align}
   \ket{h_x}\coloneqq\left(\frac{1}{\sqrt{m}}\sum_{i=1}^m
   \ket{i}\ket{E_i(x)}\right)^{\otimes r},
\end{align}
where $E_i(x)$ is the $i$th bit of $E(x)$.
Then 
\begin{align}
|\langle h_x|h_{x'}\rangle|
=\left(\frac{1}{m}\sum_{i=1}^m\langle E_i(x)|E_i(x')\rangle\right)^r
\le\delta^r
< \left(\frac{11}{12}\right)^{\frac{\log\eta}{\log\frac{11}{12}}}  =\eta.
\end{align}
The number $N$ of qubits of $\ket{h_x}$ is
\begin{align}
N=r(\lceil\log(m)\rceil+1)
\le \left(\frac{\log\eta^{-1}}{\log\frac{12}{11}}+1\right)(\log(c\ell)+2)<(8\log\eta^{-1}+1)(\log\ell+4).
\end{align}
Therefore we have $N=O(\log\ell\log\eta^{-1})$.

\end{proof}

\begin{remark}\label{rem:QROM} 
The work \cite[Section 3]{fingerprinting} also shows an alternative non-constructive method for quantum fingerprinting with better output length.  
    With some rephrasing, 
    they show that 
    we can take $m=\log \ell+O(\log \eta^{-1})$ in such a way that 
    if we define 
    $\ket{h_x}:=2^{-m/2}\sum_{z\in \bit^m}(-1)^{H(x||z)}\ket{z}$
    for a random function $H:\bit^{\ell+m}\rightarrow \bit$, then 
    for $(1-2^{-\Omega(\ell)})$-fraction of $H$, 
    we have 
    $|\langle h_x|h_{x'}\rangle|\le\eta$ for all pairs of distinct $\ell$-bit strings $x\neq x'$.  
    Then they argue that there exists a fixed $H$ which gives a quantum fingerprinting with the desired parameter by a non-constructive argument. 
    We observe that the above construction can be directly used as quantum fingerprinting in the quantum random oracle model (QROM)~\cite{AC:BDFLSZ11}. 
    This means that we have a quantum fingerprinting with output length $\log \ell+O(\log \eta^{-1})$ in the QROM, which is better than the one given in  \Cref{lem:finger}. We do not know how to achieve the same bound in the plain model (i.e., without relying on random oracles). 
    The improvement from $O(\log(\ell)\log(\eta^{-1}))$ to $\log \ell+O(\log \eta^{-1})$ does not lead to any asymptotical improvement for output length of $\epsilon$-OWSGs 
    in the parameter regimes which we consider in this paper.
\end{remark}

\takashi{I added the following lemma and its corollary.}

\begin{lemma}[{\cite[Lemma 4]{Rastegin_2007}}]\label{lem:TD_generalized}
   For any Hermitian matrices $A$ and $B$, it holds that 
    \begin{align*}
   \max_{\Pi}\Tr(\Pi(A-B))= \|A-B\|_{tr}+\frac{\Tr(A-B)}{2}
    \end{align*}
    where the maximum is taken over all projections $\Pi$.
\end{lemma}

\begin{corollary}\label{cor:trace_of_projection_and_trace_distance}
For any Hermitian matrices $A$ and $B$ and any projection $\Pi$, it holds that 
\begin{align*}
\left|\Tr\left(\Pi(A-B)\right)\right|\le 2\|A-B\|_{tr}.
\end{align*}
\end{corollary}
\begin{proof}
It immediately follows from \Cref{lem:TD_generalized} 
and $\Tr(M)\le \|M\|_{1} = 2\|M\|_{tr}$ for any Hermitian matrix $M$. 
\end{proof}

\begin{lemma}
\label{lem:fidelity_upperbound} 
Let $\{\ket{\psi_k}\}_{k\in \bit^n}$ be any set of pure states.
Then
\begin{align}
F\left(\frac{1}{2^n}\sum_{k\in\bit^n}\ket{\psi_k}\bra{\psi_k},\frac{I^{\otimes m}}{2^m}\right)
\le \frac{2^n}{2^m}.
\end{align}
\end{lemma}
\begin{proof}
Let $\sum_{k\in\bit^n}\alpha_k\ket{\phi_k}\bra{\phi_k}$ be a diagonalization of
$\frac{1}{2^n}\sum_{k\in\bit^n}\ket{\psi_k}\bra{\psi_k}$.
Then
\begin{align}
F\left(\frac{1}{2^n}\sum_{k\in\bit^n}\ket{\psi_k}\bra{\psi_k},\frac{I^{\otimes m}}{2^m}\right)
=\left(\sum_k \sqrt{\frac{\alpha_k}{2^m}}\right)^2
\le(\sum_k \alpha_k)\left(\sum_k \frac{1}{2^m}\right)
= \frac{2^n}{2^m}.
\end{align}
Here in the inequality we have used 
Cauchy–Schwarz inequality.
\end{proof}
\section{Lower Bounds for ($\epsilon$-)OWSGs}\label{sec:lower_bound}
In this section, we prove lower bounds for output length of ($\epsilon$-)OWSGs (\Cref{thm:intro_impossibility_log_OWSG,thm:intro_impossibility_inverse-poly_OWSG,thm:intro_impossibility_constant_weak_OWSG,thm:intro_impossibility_loglog_constant_OWSG}).

\subsection{Lower Bounds for OWSGs and $\secp^{-\Theta(1)}$-OWSGs}
\begin{theorem}[Restatement of \Cref{thm:intro_impossibility_log_OWSG}]\label{thm:impossibility_log_OWSG}
There do not exist OWSGs with $O(\log \secp)$-qubit outputs.
\end{theorem}
\begin{proof}
    Suppose that QPT algorithms
$(\KeyGen,\StateGen,\Ver)$ satisfy the syntax and correctness of OWSGs and has 
$m=O(\log \secp)$-qubit outputs.  
We first prove the following:
\begin{align}\label{eq:upper_bound_expected_TD}
\mathbb{E}_{k,k'}\|\phi_k-\phi_{k'}\|_{tr}\le \sqrt{1-2^{-m}}.
\end{align}
We prove this below. 
For each $k$, we consider the diagonalization of $\phi_k$ to write 
\begin{align*}
    \phi_k=\sum_{i}p_{k,i}|\phi_{k,i}\rangle\langle\phi_{k,i}|.
\end{align*}
Then we have 
\begin{align}
\mathbb{E}_{k,k'}\|\phi_k-\phi_{k'}\|_{tr}
&=
\mathbb{E}_{k,k'}\left\|\sum_{i}p_{k,i}|\phi_{k,i}\rangle\langle\phi_{k,i}|-\sum_{j}p_{k',j}|\phi_{k',j}\rangle\langle\phi_{k',j}|\right\|_{tr}\\
&= \mathbb{E}_{k,k'}\left\|\sum_{i,j}p_{k,i}p_{k',j}\left(|\phi_{k,i}\rangle\langle\phi_{k,i}|-|\phi_{k',j}\rangle\langle\phi_{k',j}|\right)\right\|_{tr} \label{eq:sum_is_one}\\
&\le \mathbb{E}_{k,k'}\sum_{i,j}p_{k,i}p_{k',j}\left\||\phi_{k,i}\rangle\langle\phi_{k,i}|-|\phi_{k',j}\rangle\langle\phi_{k',j}|\right\|_{tr} \label{eq:triangle}\\
&= \mathbb{E}_{k,k'}\sum_{i,j}p_{k,i}p_{k',j}\sqrt{1-|\langle \phi_{k,i}|\phi_{k',j}\rangle|^2}\\
&\le  \sqrt{1-\mathbb{E}_{k,k'}\sum_{i,j}p_{k,i}p_{k',j}|\langle \phi_{k,i}|\phi_{k',j}\rangle|^2} \label{eq:concave}\\
&\le \sqrt{1-2^{-m}}, \label{eq:welch}
\end{align}
where 
\Cref{{eq:sum_is_one}} follows from $\sum_{i}p_{k,i}=\sum_{j}p_{k'.j}=1$ for each $k.k'$, 
\Cref{eq:triangle} follows from the triangle inequality for trace distance, 
\Cref{eq:concave} follows from the concavity of the function $\sqrt{1-x}$, and 
\cref{eq:welch} follows from \Cref{cor:Welch}. 
This completes the proof of \Cref{eq:upper_bound_expected_TD}. 

Then we consider a trivial adversary $\cA$ against its one-wayness that ignores its input (except for the security parameter) and simply outputs a fresh random key $k'\gets \KeyGen(1^\secp)$. 
Then we have 
\begin{align}
    \Pr[\cA~\text{wins}]
    &=\mathbb{E}_{k,k'}\Pr[\Ver(k',\phi_k)=\top]\\
    &\ge \mathbb{E}_{k,k'}\left[\Pr[\Ver(k,\phi_k)=\top]-\|\phi_k-\phi_{k'}\|_{tr}\right]\\
    &\ge 1-\negl(\secp)-\sqrt{1-2^{-m}}, \label{eq:final_lower_bound}
\end{align}
where \Cref{eq:final_lower_bound} follows from the correctness of $(\KeyGen,\StateGen,\Ver)$ and \Cref{eq:upper_bound_expected_TD}. 
Note that we have $\sqrt{1-x}\le 1-x/2$ for all $0\le x \le 1$.
\takashi{In the previous version, I was using a looser bound  $\sqrt{1-x}\le 1-x/4$, but I found we actually have $\sqrt{1-x}\le 1-x/2$. 
(The Taylor expansion of $\sqrt{1-x}$ looks like $1-x/2-x^2/8-...$. I misunderstood it as $1-x/2+x^2/8-...$.)
}
Therefore we have  
\begin{align}
    \Pr[\cA~\text{wins}]\ge 
    2^{-m-1} -\negl(\secp).    
\end{align}
Since $m=O(\log \secp)$, the RHS is non-negligible and thus $\cA$ succeeds in breaking one-wayness. 
\end{proof}
\begin{remark}
In the above proof, $\cA$ does not use the target state at all. This means that even \emph{zero-copy} OWSGs (where the adversary is given no copy of the target state) with $O(\log \secp)$-qubit outputs do not exist. 
\end{remark}


\takashi{I added the following theorem and proof.}
The above proof can be easily extended to give a lower bound for $\secp^{-\Theta(1)}$-OWSGs. 

\begin{theorem}[Restatement of \Cref{thm:intro_impossibility_inverse-poly_OWSG}]\label{thm:impossibility_inverse-poly_OWSG}
For any constant $c\in \mathbb{N}$,   
there do not exist $\secp^{-c}$-OWSGs of which output length is at most $c \log\secp-2$.    
\end{theorem}
\begin{proof}
Suppose that the output length is $m< c' \log\secp-2$. By the proof of \Cref{thm:impossibility_log_OWSG}, there is a QPT adversary that wins with probability at least $2^{-m-1}-\negl(\secp)\ge 2\secp^{-c}-\negl(\secp)\ge \secp^{-c}$ for all sufficiently large $\secp$. 
This means that it is not a $\secp^{-c}$-OWSG. 
\end{proof}
\begin{remark}
    As one can easily see from the proof, we can improve the lower bound from $c \log\secp-2$ to $c \log\secp-1-\epsilon$ for any constant $\epsilon>0$. 
\end{remark}

\subsection{Lower Bounds for $\Theta(1)$-OWSGs and Weak OWSGs.}
For proving \Cref{thm:intro_impossibility_constant_weak_OWSG,thm:intro_impossibility_loglog_constant_OWSG}, we prove the following theorem. 

\begin{theorem}\label{thm:weak_OWSG_attack}
    Let $\Delta=\Delta(\lambda)>0$ and $m=m(\lambda)\in \mathbb{N}$ and $C$ be a constant taken from~\Cref{lem:epsnet}.
    If $(6C/\Delta)^{2^{2m}} = \poly(\lambda)$, there do not exist $(1-\Delta)$-OWSGs with $m$-qubit outputs.
\end{theorem}
\begin{proof}
    Suppose that QPT algorithms $(\KeyGen,\StateGen,\Ver)$ satisfy the syntax and correctness of $(1-\Delta)$-OWSGs with $m$-qubit outputs. 
    Let $\gamma=\Delta/6$, $d=2^m$, $t=144\lambda d^4/\gamma^2$, and $T= \frac{(C/\gamma)^{d^2} \log \lambda}{\gamma}$. 
    We consider the following adversary $\cA(1^\lambda,\phi_k^{\otimes t})$:
    \begin{enumerate}
        \item Construct an $\gamma$-net $\enet(d,\gamma)$.
        \item Given $\phi_k^{\otimes t}$, run $M \gets \tomo(\phi_k^{\otimes t},\gamma)$.
        \item Construct a set $\cN = \{ N \in \enet(d,\gamma): \|N-M\|_{tr} \le \gamma\}.$
        \item For each $i \in [T]$, do:
        \begin{enumerate}
            \item Sample $k_i \gets \KeyGen(1^\lambda)$ and run $M_i \gets \tomo(\phi_{k_i}^{\otimes t},\gamma).$ 
            \item If there exists $N \in \cN$ such that $\|M_i - N \|_{tr} \le \gamma$, output $k':=k_i$ and terminate.
        \end{enumerate}
        \item Output $\bot$.
    \end{enumerate}
    Because of~\Cref{lem:tomography,lem:epsnet}, the time complexity of this algorithm is bounded by a polynomial in $1/\gamma,d,\lambda,T$ and $|\enet(d,\gamma)$. 
    By the condition, $|\enet(d,\gamma)|\le (C/\gamma)^{d^2} = (6C/\Delta)^{2^{2m}}=\poly(\lambda)$ and this implies that $1/\Delta=\poly(\lambda)$ and $d=2^m=\poly(\lambda)$. This proves that $\cA$ is a QPT algorithm.

    To prove the correctness, we first argue that $\cA$ outputs $\bot$ with a small probability. 
    Let $\epsilon_{bad} := \frac{\gamma}{|\enet(d,\gamma)|}$.
    Let $\cB$ be the set of elements $N \in \enet(d,\gamma)$ such that the following inequality holds:
    \begin{align}
        \Pr\left[
            \| N-M' \|_{tr} \le \gamma:
            k' \gets \KeyGen(1^\lambda), M' \gets \tomo(\phi_{k'}^{\otimes t},\gamma)
        \right] \le \epsilon_{bad}.
    \end{align}
    
    Let $\sf{bad}$ be the event that $\cN \cap \cB\neq \emptyset$.
    It holds that
    \begin{align}
        \Pr[\sf{bad}]=\Pr[\cN \cap \cB \neq \emptyset] 
        \le \sum_{N \in \cB} \Pr[N \in \cN]
        \le |\cB| \cdot \epsilon_{bad} \le \gamma,
    \end{align}
    where we have used $|\cB| \le |\enet(d,\gamma)|$ at the last inequality.

    Suppose that $\sf{bad}$ does not occur. Then, since $\cN$ is not empty, there exists an element $N \in \cN$ such that 
    \begin{align}
        \Pr\left[
            \| N-M' \|_{tr} \le \gamma:
            k' \gets \KeyGen(1^\lambda), M' \gets \tomo(\phi_{k'}^{\otimes t},\gamma)
        \right] > \epsilon_{bad},
    \end{align}
    which means that each iteration of the fourth step terminates with probability at least $\frac{\gamma}{|\enet(d,\gamma)|}$. 
    Then we have
    \begin{align}
        &
        \Pr\left[
            k'\neq \bot: k \gets \KeyGen(1^\lambda),k' \gets \cA(1^\secp,\phi_k^{\otimes t})
        \right] 
        \\
        &
        \ge\Pr\left[
            k'\neq \bot: k \gets \KeyGen(1^\lambda),k' \gets \cA(1^\secp,\phi_k^{\otimes t}) | \lnot {\sf{bad}}
        \right] \cdot \Pr[\lnot {\sf{bad}}]
        \\
        &\label{eq:indep}
        \ge \left(1- \left( 1 -  \epsilon_{bad}\right)^T\right) \cdot (1-\gamma)
        \ge 1-\gamma -\negl(\lambda),
    \end{align}
    where the first inequality in \Cref{eq:indep} holds because the iterations are independent, and the last inequality holds because of the choice of $T \ge {\log \lambda}/{\epsilon_{bad}}$.

    Finally, we prove that the output of $\cA$ is accepted by $\Ver$ with high probability whenever it is not $\bot$. If $\cA$ outputs $k_i$, $\| \phi_{k} - \phi_{k_i} \|_{tr} \le 4\gamma$ with probability $1-\negl(\lambda)$. 
    By the triangular inequality, we have
    \begin{align}
        \| \phi_{k} - \phi_{k_i} \|_{tr} \le 
        \|\phi_k -M \|_{tr} + \|M-N \|_{tr}  + \|N-M_i \|_{tr} + \|M_i-\phi_{k_i} \|_{tr},
    \end{align}
    all of which are bounded by $\gamma$ with probability $1-\negl(\lambda)$ because of~\Cref{lem:tomography} and the definition of $\enet(d,\gamma)$.
    This implies that
    \begin{align}\label{eq:weak_OWSG_correct}
        \Pr[\top \gets \Ver(k_i,\phi_k)] \ge \Pr[\top \gets \Ver(k,\phi_k)] - \| \phi_{k} - \phi_{k_i} \|_{tr} \ge 1- 4\gamma -\negl(\lambda)
    \end{align}
    holds with probability $1-\negl(\lambda)$,
    where the last inequality follows from the correctness of $(\KeyGen,\StateGen,\Ver)$.
    Using $\gamma=\Delta/6$, the advantage of $\cA$ is bounded as follows:
    \begin{align}
        \Pr[\cA\text{ wins}] &= \Pr[\top \gets \Ver(k',\phi_k): k\gets \KeyGen(1^\lambda), k' \gets \cA(1^\secp,\phi_k^{\otimes t})]\\
        &\ge (1- 4\gamma -\negl(\lambda)) \cdot (1-\gamma -\negl(\lambda)) \cdot (1-\negl(\lambda)\label{eqn:probs}\\
        &\ge 1-5\gamma -\negl(\lambda) \ge 1-\Delta
    \end{align}
    where the first term of \Cref{eqn:probs} is from~\Cref{eq:weak_OWSG_correct}, the second term is from~\Cref{eq:indep}, and the last term is the probability that \Cref{eq:weak_OWSG_correct} holds.
    Therefore, $\cA$ is a QPT algorithm breaking $(1-\Delta)$-OWSG of $(\KeyGen,\StateGen,\Ver)$.
\end{proof}
\Cref{thm:intro_impossibility_constant_weak_OWSG,thm:intro_impossibility_loglog_constant_OWSG} follow from \Cref{thm:weak_OWSG_attack} 
as immediate corollaries. 

\begin{corollary}[Restatement of \Cref{thm:intro_impossibility_loglog_constant_OWSG}]
    There do not exist $\epsilon$-OWSGs with $m$-qubit outputs for any constant $\epsilon>0$ and $m=\frac{1}{2}\log \log \lambda + O(1)$. 
\end{corollary}
\begin{proof}
    If the output length $m$ is $m=\frac{1}{2}\log\log \lambda+O(1)$ so that $2^{2m} = O(\log \lambda),$ 
    then for any constant $\epsilon>0$, $(6C/(1-\epsilon))^{2^{2m}}=(6C/(1-\epsilon))^{O(\log \lambda)} = \poly(\lambda)$ and we can apply~\Cref{thm:weak_OWSG_attack}. 
\end{proof}
\begin{remark}
    For the pure case, we can prove that there do not exist $\epsilon$-OWSGs with $(\log\log \lambda + O(1))$-qubit outputs for any constant $\epsilon>0$ using the pure case of~\Cref{lem:epsnet}. 
\end{remark}
\begin{corollary}[Restatement of \Cref{thm:intro_impossibility_constant_weak_OWSG}]
    There do not exist weak OWSGs with any constant output length.
\end{corollary}
\begin{proof}
        Consider the security of weak OWSGs with the constant output length $m$, which is $(1-\epsilon)$-secure for some inverse polynomial $\epsilon$. Since $m$ is constant, $(6C/\epsilon)^{2^{2m}}=\poly(\lambda)$, thus \Cref{thm:weak_OWSG_attack} applies.
\end{proof}

\section{Upper Bounds for ($\epsilon$-)OWSGs}\label{sec:upper_bound}
In this section, we prove upper bounds for output length of ($\epsilon$-)OWSGs (\Cref{thm:intro_possibility_inverse-poly_OWSG,thm:intro_possibility_Ologlog_constant_OWSG,thm:intro_possibility_superconstant_weak_OWSG}).

\subsection{Upper Bounds for OWSGs and $\secp^{-\Theta(1)}$-OWSGs}
We first observe that the combination of previous results already 
gives OWSGs with slightly super-logarithmic-qubit outputs. 
\begin{theorem}[\cite{C:BraShm20,CGGHLP23}]\label{thm:OWSG_with_slightly_poly}
If there exist OWFs, then there exist OWSGs with $m(\secp)$-qubit pure state outputs for any QPT-computable function $m(\secp)=\omega(\log\secp)$.
\end{theorem}

\begin{proof}
According to \cite{C:BraShm20}, if there exist OWFs, then there exist PRSGs with input length $\secp$ and any output length $m$. 
From \cite[Theorem 3.4]{CGGHLP23}, these PRSGs with $m=\omega(\log \secp)$ are also
OWSGs.
\end{proof}

We show that the result of \cite{CGGHLP23} can be adapted to construct $1/\poly(\secp)$-OWSGs with $O(\log \secp)$-qubit outputs. 
We recall the following lemma from~\cite[Lemma 3.2]{CGGHLP23} with a small modification.
\begin{lemma}\label{lem:PRSGtoOWSG}
    Let $(\KeyGen,\StateGen)$ be an $\epsilon$-PRSG, where $\StateGen$ maps $n$-bit strings to $m$-qubit pure states. 
    Let $h \in [0,1]$ and $\delta= 2^n \cdot \left(1-h\right)^{2^m-1} + h$.
    Define a QPT algorithm $\Ver(k',\ket{\phi_k}) \to \top/\bot$ to be the projection of $\ket{\phi_k}$ onto $\ket{\phi_{k'}}$ with the output $\top$ only when the projection is successful. 
    Then $(\KeyGen,\StateGen,\Ver)$ is an $(\epsilon+\delta)$-OWSG.
\end{lemma}
\begin{proof}[Proof sketch of \Cref{lem:PRSGtoOWSG}]
    Most parts of the proof are identical to the original one.
    When proving the security, we recall the following concentration inequality for Haar random states instead of~\cite[Lemma 3.1]{CGGHLP23}. 
    \begin{lemma}\label{lem:Haar_concentration}
        For any $h \in [0,1]$ and for any $m$-qubit quantum state $\phi_0$, it holds that for the Haar measure on $m$-qubit quantum states $\mu_m$,
        \[
            \Pr_{\ket{\psi} \gets \mu_m} \left[ 
                |\braket{\psi|\phi_0}|^2 \ge h 
            \right] 
            = (1-h)^{2^m-1}.
        \]
    \end{lemma}
    The proof of \Cref{lem:Haar_concentration} can be found in the proof of~\cite[Lemma 3.6]{AK07}. Using \Cref{lem:Haar_concentration} instead of~\cite[Lemma 3.1]{CGGHLP23} gives the choice $\delta=2^n \cdot (1-h)^{2^m-1}+h$, where $h$ has a similar role to $1/f(n)$ in the original proof.
\end{proof}

\begin{theorem}[Restatement of \Cref{thm:intro_possibility_inverse-poly_OWSG}]\label{thm:possibility_inverse-poly_OWSG}
   For any constant $c\in \mathbb{N}$, there exist $\lambda^{-c}$-OWSGs with $(c+1) \log \lambda+O(1)$-qubit pure state outputs assuming the existence of OWFs.
\end{theorem} 
\begin{proof}
    By~\cite{C:BraShm20}, there exist PRSGs with input length $\lambda$ and any output length $m$ assuming the existence of OWFs.
    We apply~\Cref{lem:PRSGtoOWSG} with $h=\lambda^{-c}/4$ and
    \begin{align}
        m\ge 
        \log \left(
            \frac{\log_e\left(
                2^{\lambda+2}\lambda^c
            \right)}{h}+1
        \right)~\Longrightarrow~2^n \cdot (1-h)^{2^m-1} \le \lambda^{-c}/4,
    \end{align}
    which gives $\delta=\lambda^{-c}/2$.
    To prove the asymptotic bound of $m(\lambda)$, we choose a constant $K$ such that $2^{\lambda+2}\lambda^c \le e^\lambda$ for all $\lambda\ge K$. In this case, the minimal $m$ satisfies
    \begin{align}
        m \ge \log \left(
            4 \lambda^{c+1} + 1
        \right)
        \ge (c+1) \log \lambda + 3,
    \end{align}
    thus we can choose $m(\lambda)=(c+1)\log \lambda +O(1)$.
\end{proof}

\subsection{Upper Bound for $\Theta(1)$-OWSGs}
For proving \Cref{thm:intro_possibility_Ologlog_constant_OWSG}, 
we prove the following theorem.  
\begin{theorem}
\label{thm:upperbound}
If there exist OWFs that map $\secp$-bit inputs to $\ell(\secp)$-bit outputs,   
then there exist $(\eta(\secp)^2+\negl(\secp))$-OWSGs with $O(\log(\ell(\secp))\log(\eta(\secp)^{-1}))$-qubit pure state outputs
for any QPT-computable function $\eta(\secp)$. 
\end{theorem}

\begin{proof}
For notational simplicity, we simply write $\ell$ and $\eta$ to mean  $\ell(\secp)$ and $\eta(\secp)$, respecitively. 
Let 
$f$ be a OWF that maps $\secp$-bit inputs to $\ell$-bit outputs. 
We construct an $(\eta^2+\negl(\secp))$-OWSG as follows:
\begin{itemize}
\item
$\KeyGen(1^\secp)\to k:$
Sample $k\gets\bit^\secp$.
    \item 
    $\StateGen(k)\to \ket{\phi_k}:$ 
    Output $\ket{\phi_k}\coloneqq\ket{h_{f(k)}}$, where $\ket{h_{f(k)}}$ is the quantum fingerprinting state of $f(k)$ given in \cref{lem:finger}. 
    \item
    $\Ver(k',\ket{\phi_k})\to\top/\bot:$ 
    Project $\ket{\phi_k}$ onto $\ket{\phi_{k'}}$.
    If the projection is successful, output $\top$. Otherwise, output $\bot$. 
\end{itemize}
Let $\cA$ be a QPT adversary of this OWSG that receives $t$ copies of $\ket{\phi_k}$.
Then, as we will show below, there exists a QPT adversary $\cB$ of the OWF $f$ that satisfies\footnote{We say that $\cA$ wins if $\Ver(k',\ket{\phi_k})\rightarrow \top$ where 
$k\gets \KeyGen(1^\secp)$ and $k'\gets \cA(1^\secp,\ket{\phi_k}^{\otimes t})$.
Similarly, we say that $\cB$ wins if $f(x')=f(x)$ where $x\gets \bit^\secp$ and $x'\gets \cB(1^\secp,f(x))$. 
} 
\begin{align}
\Pr[\cA~\mbox{wins}]\le \Pr[\cB~\mbox{wins}]+\eta^2.    
\label{relation}
\end{align}
Because $f$ is a OWF, 
we have
\begin{align}
\Pr[\cA~\mbox{wins}]\le \negl(\secp)+\eta^2,    
\end{align}
which shows that the OWSG is $(\eta^2+\negl(\secp))$-secure.
From \cref{lem:finger},
the output length of the OWSG is at most $O(\log(\ell)\log(\eta^{-1}))$.

Now we show \cref{relation}.
Let $\cA$ be a QPT adversary of the OWSG.
From this $\cA$, we construct an adversary $\cB$ of the OWF $f$ as follows:
\begin{enumerate}
    \item 
    Given $(1^\secp,f(k))$, 
    run $k'\gets\cA(1^\secp,\ket{h_{f(k)}}^{\otimes t})$.
    \item 
    Output $k'$.
\end{enumerate}
From \cref{lem:finger}, $|\langle h_{f(k)}|h_{f(k')}\rangle|\le\eta$ if $f(k)\neq f(k')$.
Therefore we have
\begin{align}
\Pr[\cA~\mbox{wins}]
&=\frac{1}{2^\secp}
\sum_{k}
\sum_{k'}
\Pr[k'\gets\cA(1^\secp,\ket{h_{f(k)}}^{\otimes t})]
|\langle h_{f(k)}|h_{f(k')}\rangle|^2\\
&=\frac{1}{2^\secp}
\sum_{k}
\sum_{k':f(k')=f(k)}
\Pr[k'\gets\cA(1^\secp,\ket{h_{f(k)}}^{\otimes t})]
|\langle h_{f(k)}|h_{f(k')}\rangle|^2\\
&+\frac{1}{2^\secp}
\sum_{k}
\sum_{k':f(k')\neq f(k)}
\Pr[k'\gets\cA(1^\secp,\ket{h_{f(k)}}^{\otimes t})]
|\langle h_{f(k)}|h_{f(k')}\rangle|^2\\
&\le\frac{1}{2^\secp}
\sum_{k}
\sum_{k':f(k')=f(k)}
\Pr[k'\gets\cA(1^\secp,\ket{h_{f(k)}}^{\otimes t})]\\
&+\frac{1}{2^\secp}
\sum_{k}
\sum_{k':f(k')\neq f(k)}
\Pr[k'\gets\cA(1^\secp,\ket{h_{f(k)}}^{\otimes t})]
\eta^2\\
&\le\frac{1}{2^\secp}
\sum_{k}
\sum_{k':f(k')=f(k)}
\Pr[k'\gets\cA(1^\secp,\ket{h_{f(k)}}^{\otimes t})]
+\eta^2\\
&\le\Pr[\cB~\mbox{wins}]
+\eta^2.
\end{align}
\end{proof}

\Cref{thm:intro_possibility_Ologlog_constant_OWSG} follows from 
\Cref{thm:upperbound} as an immediate corollary. 
\begin{corollary}[Restatement of \Cref{thm:intro_possibility_Ologlog_constant_OWSG}]\label{cor:possibility_Ologlog_constant_OWSG}
    For any constant $\epsilon>0$, there exist $\epsilon$-OWSGs with $O(\log\log \secp)$-qubit pure state outputs assuming the existence of subexponentially secure OWFs.   
\end{corollary}
\begin{proof}
    Without loss of generality, we can assume that a subexponentially secure OWF is length-preserving. i.e., it maps $\secp$-bit inputs to $\secp$-bit outputs. 
    Let $f$ be a length-preserving subexponentially secure OWFs. 
    Since $f$ is subexponentially secure there is a constant $0<c < 1$ such that
for any QPT adversary $\cA$, 
\begin{equation}
\Pr[f(x')=f(x):
x\gets\bit^\secp,
x'\gets\cA(1^\secp,f(x))
]
\le 2^{-\secp^{c}}
\end{equation}
for all sufficiently large $\secp$. 
Then we construct a OWF $g$ that maps $\secp$-bit inputs to $\ell(\secp):=\lfloor(\log \secp)^{2c^{-1}}\rfloor$-bit outputs as follows: 
for any $x\in \bit^\secp$, 
$g(x):=f(x_{[\ell]})$ where $x_{[\ell]}$ is the first $\ell$ bits of $x$. 
By a straightforward reduction, for any QPT adversary $\cA$, we have 
\begin{equation}
\Pr[g(x')=g(x):
x\gets\bit^\secp,
x'\gets\cA(1^\secp,g(x))
]
\le 2^{-\ell^{c}}\le 2^{-(\log \secp)^2}=\negl(\secp).
\end{equation}
Thus, $g$ is a OWF that maps $\secp$-bit inputs to $\ell(\secp)=\lfloor(\log \secp)^{2c^{-1}}\rfloor$-bit outputs.  
By applying \Cref{thm:upperbound} for $\eta:=(\epsilon/2)^{1/2}$, there is a $(\epsilon/2+\negl(\secp))$-OWSG with $O(\log \ell)=O(\log \log \secp)$-qubit pure state outputs. 
Since $\epsilon/2+\negl(\secp)<\epsilon$ for sufficiently large $\secp$, this implies \Cref{cor:possibility_Ologlog_constant_OWSG}. 
\end{proof}

\subsection{Upper Bound for Weak OWSGs}
For proving \Cref{thm:intro_possibility_superconstant_weak_OWSG}, we prove the following theorem. 
\begin{theorem}
\label{thm:upperbound_weak}
If there exist OWFs that map $\secp$-bit inputs to $\ell(\secp)$-bit outputs,   
then there exist 
weak 
OWSGs with $\lceil \frac{\ell(\secp)}{\log \lambda} \rceil$-qubit outputs.
\end{theorem}
\begin{proof} 
We simply write $\ell$ to mean $\ell(\secp)$ for notational simplicity. 
Let $f$ be a  OWF that maps $\secp$-bit inputs to $\ell$-bit outputs.  
Define $t:=\lceil \frac{\ell}{\log \lambda} \rceil$, and we embed $\bit^\ell$ into $[\secp]^t$. 
Then we can regard $f$ as a function that maps $\secp$-bit strings to an element of $[\secp]^t$. 
    For $\theta$, let 
    \begin{align*}
        \ket{+_\theta}:=\frac{1}{\sqrt{2}}(\ket{0}+e^{i\theta}\ket{1}).
    \end{align*}
    We construct a OWSG as follows:
\begin{itemize}
\item
$\KeyGen(1^\secp)\to k:$
Sample $k\gets\bit^\secp$.
    \item 
    $\StateGen(k)\to \ket{\phi_k}:$ 
    Compute $f(k)=(y_1,y_2,...,y_t)\in [\secp]^t$ and 
    output $\ket{\phi_k}\coloneqq\bigotimes_{j\in [t]}\ket{+_{2\pi y_j/\secp}}$. 
    \item
    $\Ver(k',\ket{\phi_k})\to\top/\bot:$ 
    Project $\ket{\phi_k}$ onto $\ket{\phi_{k'}}$.
    If the projection is successful, output $\top$. Otherwise, output $\bot$. 
\end{itemize}
For any $y\neq y'$, we have $|\langle +_{2\pi y/\secp}|+_{2\pi y'/\secp}\rangle|\le \sqrt{\frac{1+\cos (2\pi/\secp)}{2}}$.
This is because
\begin{align}
|\langle +_{2\pi y/\sep}|+_{2\pi y'/\secp}\rangle|
&=\left|\frac{1+e^{i2\pi(y-y')/\secp}}{2}\right|\\
&=\left|\frac{1+\cos(2\pi(y-y')/\secp)+i\sin(2\pi(y-y'))}{2}\right|\\
&=\sqrt{
\left(\frac{1+\cos(2\pi(y-y')/\secp)}{2}\right)^2+\frac{\sin^2(2\pi(y-y'))}{4}
}\\
&\le \sqrt{\frac{1+\cos (2\pi(y-y')/\secp)}{2}}\\
&\le \sqrt{\frac{1+\cos (2\pi/\secp)}{2}}.
\end{align}
Thus, for any $k,k'$ such that $f(k)\ne f(k')$, 
we have 
$|\langle \phi_k|\phi_{k'}\rangle|\le \sqrt{\frac{1+\cos (2\pi/\secp)}{2}}$. 
By a similar reduction to that in the proof of \Cref{thm:upperbound}, we can see that the above OWSG satisfies 
$\left(\frac{1+\cos (2\pi/\secp)}{2}+\negl(\secp)\right)$-one-wayness. 
Since $\cos \theta= 1-\theta^2/2+O(\theta^4)$, 
 we have $\cos(2\pi/\secp)\le 1-(\pi/\secp)^2$ for sufficiently large $\secp$.   
Thus, $\frac{1+\cos (2\pi/\secp)}{2}=1-\frac{1}{\poly(\secp)}$ and thus the above is a weak OWSG. 
\end{proof}
\Cref{thm:intro_possibility_superconstant_weak_OWSG} follows from \Cref{thm:upperbound_weak} as an immediate corollary.
\begin{corollary}[Restatement of \Cref{thm:intro_possibility_superconstant_weak_OWSG}]\label{cor:possibility_superconstant_weak_OWSG}
For any QPT-computable function $m(\secp)=\omega(1)$, there exist weak OWSGs with $m(\secp)$-qubit pure state outputs assuming the existence of exponentially secure  OWFs that map $\secp$-bit inputs to $O(\secp)$-bit outputs. 
\end{corollary}
\begin{proof}
    Suppose that $f$ is an exponentially secure OWF that maps $\secp$-bit inputs to $k\secp$-bit outputs for some constant $k\ge 1$. 
    Since it is exponentially secure,  there is a constant $0<c < 1$ such that
for any QPT adversary $\cA$, 
\begin{equation}
\Pr[f(x')=f(x):
x\gets\bit^\secp,
x'\gets\cA(1^\secp,f(x))
]
\le 2^{-c \secp}
\end{equation}
for all sufficiently large $\secp$. 
Then we construct a OWF $g$ that maps $\secp$-bit inputs to $\ell(\secp):=\lfloor m(\secp)\log \secp \rfloor$-bit outputs as follows: 
for any $x\in \bit^\secp$, 
$g(x):=f(x_{[\ell]})$ where $x_{[\ell]}$ is the first $\ell$ bits of $x$. 
By a straightforward reduction, for any QPT adversary $\cA$, we have 
\begin{equation}
\Pr[g(x')=g(x):
x\gets\bit^\secp,
x'\gets\cA(1^\secp,g(x))
]
\le 2^{-c\ell}\le 2^c\cdot \secp^{-cm(\secp)}=\negl(\secp),
\end{equation}
where the final equality follows from $m(\secp)=\omega(1)$.   
Thus, $g$ is a OWF that maps $\secp$-bit inputs to $\ell(\secp)=\lfloor m(\secp)\log \secp \rfloor$-bit outputs.  
By applying \Cref{thm:upperbound_weak}, there is a weak OWSG with $\lceil \frac{\ell(\secp)}{\log \lambda} \rceil$-qubit outputs. 
By the definition of $\ell(\secp)$, we have 
$\lceil \frac{\ell(\secp)}{\log \lambda} \rceil\le m(\secp)$.  
Since we can arbitrarily increase the output length of weak OWSGs by a trivial padding, this implies \Cref{cor:possibility_superconstant_weak_OWSG}. 
\end{proof}
\section{EFIs}\label{sec:EFI}
In this section, we show upper and lower bounds for the output length of EFIs. 
\subsection{Lower Bound}
\begin{theorem}[Restatement of \Cref{thm:intro_lower_EFI}]\label{thm:lower_EFI}
There do not exist $O(\log\secp)$-qubit EFIs.
\end{theorem}
\begin{proof}
Let $\StateGen$ be a QPT algorithm that takes a bit $b$ and $1^\secp$ for $\secp\in \mathbb{N}$ as input and outputs an $\ell(\secp)$-qubit state $\rho_b$ for $\ell(\secp)=O(\log\secp)$. 
Suppose that $\rho_0$ and $\rho_1$ are statistically distinguishable, i.e.,  $\|\rho_0-\rho_1\|_{tr}\ge\frac{1}{p(\secp)}$ for some polynomial $p$. 
It suffices to show that $\rho_0$ and $\rho_1$ are not computationally indistinguishable. To show this, we consider the following adversary $\A$ that distinguishes $\rho_0$ and $\rho_1$. 
\begin{enumerate}
\item Take the security parameter $1^\secp$ and a quantum state $\sigma$, which is either of $\rho_0$ or $\rho_1$ as input.
\item Set $\delta=\delta(\secp):=\frac{1}{16p(\secp)}$.
\item \label{step:tomography}
For $b\in \bit$, generate $t:=144\secp 2^{4\ell}/\delta^2$ copies of $\rho_b$ and run $M_b\gets \tomo(\rho^{\otimes t}_b,\delta)$ where $\tomo$ is the algorithm in \Cref{lem:tomography}.
\item \label{step:diagonalize}
Compute the spectral decomposition of $M_0-M_1$: 
\begin{align*}
M_0-M_1=\sum_{i}\mu_i\ket{v_i}\bra{v_i}
\end{align*}
where $\ket{v_i}$ is an eigenvector with eigenvalue $\mu_i$ and $\{\ket{v_i}\}_i$ forms an orthonormal basis of $\mathbb{C}^{2^\ell}$. 
Define 
\begin{align*}
Q:=\sum_{i:\mu_i\ge 0}\mu_i\ket{v_i}\bra{v_i},~~~S:=\sum_{i:\mu_i< 0}-\mu_i\ket{v_i}\bra{v_i}
\end{align*}
so that $Q$ and $S$ are positive semi-definite matrices with orthogonal supports and satisfies $M_0-M_1=Q-S$. 
Let $\Pi$ be the projector onto the support of $Q$, i.e., 
\begin{align*}
\Pi:=\sum_{i:\mu_i\ge 0}\ket{v_i}\bra{v_i}.
\end{align*}
\item \label{step:Pi}
Apply projective measurement $\{\Pi,I-\Pi\}$ to $\sigma$ and output the measurement outcome. 
\end{enumerate} 

We remark that $\A$ may not be QPT since there may not exist a polynomial size (or even finite size) quantum circuit that exactly implements the projective measurement $\{\Pi,I-\Pi\}$. However, we observe that we can approximate $\A$ in QPT with an exponentially small error as follows.
First, tomography in Step \ref{step:tomography} runs in QPT since $t=144\secp 2^{4\ell}/\delta^2=\poly(\secp)$ and the running time of $\tomo$ is polynomial in $1/\delta,d=2^\ell$ and $\lambda$ as required in \Cref{lem:tomography}.  
In Step \ref{step:diagonalize}, by a numerical implementation of a standard diagonalization algorithm, we can obtain a description of a projector $\widetilde{\Pi}$ (which consists of all the entries of the matrix $\widetilde{\Pi}$) such that $\|\Pi-\widetilde{\Pi}\|_{op}\le 2^{-\Omega(\secp)}$ in QPT where $\|\cdot\|_{op}$ means the operator norm. 
It is known that any $\ell$-qubit unitary can be implemented with an error $\epsilon$ (in terms of the operator norm) by a $\poly(2^\ell,\log \epsilon^{-1})$-size quantum circuit. 
In fact, \cite[Sec. 4.5]{Nielsen_Chuang} gives a constructive way to compute the description of the approximating quantum circuit in (classical) time $\poly(2^\ell,\log \epsilon^{-1})$ given the description of the unitary. 
By setting $\epsilon:=2^{-\Omega(\secp)}$, the above implies that we can efficiently compute a description of a quantum circuit that implements projective measurement $\{\widetilde{\Pi}',I-\widetilde{\Pi}'\}$ such that 
$\|\widetilde{\Pi}-\widetilde{\Pi}'\|_{op}\le 2^{-\Omega(\secp)}$. 
By the triangle inequality, $\|\Pi-\widetilde{\Pi}'\|_{op}\le 2^{-\Omega(\secp)}$, and thus the output distribution changes only by $2^{-\Omega(\secp)}$  even if we apply $\{\widetilde{\Pi}',I-\widetilde{\Pi}'\}$ instead of $\{\Pi,I-\Pi\}$. 
This means that $\A$ can be approximated with an exponentially small error in QPT. 

Thus, it suffices to show that $\A$'s distinguishing advantage is $1/\poly(\secp)$. 
Moreover, since $\|M_b-\rho_b\|_{tr}\le \delta$ holds for $b\in \bit$ except for a negligible probability by \Cref{lem:tomography}, 
it suffices to show that $\A$'s distinguishing advantage is $1/\poly(\secp)$ 
conditioned on that $\|M_b-\rho_b\|_{tr}\le \delta$ for $b\in \bit$. We show this below.

First, by the definitions of $Q$ and $S$, we have
\begin{align}\label{eq:M_0-M_1_tr}
    \|M_0-M_1\|_{tr}=
    \frac{\Tr(Q)+\Tr(S)}{2}.
\end{align}
We also have 
\begin{align}
    \left|\frac{\Tr(Q)-\Tr(S)}{2}\right|
    &=\left|\frac{\Tr(M_0)-\Tr(M_1)}{2}\right| \label{eq:Tr_Q-S_1}\\
    &\le \frac{1}{2}\left(|\Tr(\rho_0-\rho_1)|+|\Tr(M_0-\rho_0)|+|\Tr(M_1-\rho_1)|\right) \label{eq:Tr_Q-S_2}\\
    &\le 0+\|M_0-\rho_0\|_{tr}+\|M_1-\rho_1\|_{tr} \label{eq:Tr_Q-S_3}\\
    &\le 2\delta \label{eq:Tr_Q-S_4}
\end{align}
where \Cref{eq:Tr_Q-S_1} follows from $Q-S=M_0-M_1$,  
\Cref{eq:Tr_Q-S_2} follows from the triangle inequality, 
\Cref{eq:Tr_Q-S_3} follows from $\Tr(\rho_0)=\Tr(\rho_1)=1$ and
 $\Tr(M)\le \|M\|_{1} = 2\|M\|_{tr}$ for any Hermitian matrix $M$\takashi{I changed the explanation here.}, 
and \Cref{eq:Tr_Q-S_4} follows from the assumption that  $\|M_b-\rho_b\|_{tr}\le \delta$ for $b\in \bit$. 
Then we have
\begin{align}
\Tr(\Pi(M_0-M_1))&= \Tr(Q) \label{eq:Tr_Pi_M_0-M_1-1}\\
&=\frac{\Tr(Q)+\Tr(S)}{2}+\frac{\Tr(Q)-\Tr(S)}{2} \label{eq:Tr_Pi_M_0-M_1-2}\\
&\ge \|M_0-M_1\|_{tr}-2\delta \label{eq:Tr_Pi_M_0-M_1-3}\\
&\ge \|\rho_0-\rho_1\|_{tr}
-\|M_0-\rho_0\|_{tr}
-\|M_1-\rho_1\|_{tr}
-2\delta \label{eq:Tr_Pi_M_0-M_1-4}\\
&\ge \|\rho_0-\rho_1\|_{tr}-4\delta \label{eq:Tr_Pi_M_0-M_1-5}
\end{align}
where \Cref{eq:Tr_Pi_M_0-M_1-1} follows from the definitions of $Q,S$, and $\Pi$, 
\Cref{eq:Tr_Pi_M_0-M_1-3} follows from  \Cref{eq:M_0-M_1_tr} and \Cref{eq:Tr_Q-S_1,eq:Tr_Q-S_2,eq:Tr_Q-S_3,eq:Tr_Q-S_4}, 
\Cref{eq:Tr_Pi_M_0-M_1-4} follows from the triangle inequality, 
and \Cref{eq:Tr_Pi_M_0-M_1-5} follows from the assumption that  $\|M_b-\rho_b\|_{tr}\le \delta$ for $b\in \bit$.

Then we have
\begin{align}
\Tr(\Pi\rho_0)-\Tr(\Pi\rho_1)&=
\Tr(\Pi(M_0-M_1))+
\Tr(\Pi(\rho_0-M_0))+
\Tr(\Pi(M_1-\rho_1)) \label{eq:Tr_Pi_rho_0-rho_1-1}\\
&\ge \|\rho_0-\rho_1\|_{tr}-4\delta-2\|M_0-\rho_0\|_{tr}
-2\|M_1-\rho_1\|_{tr} \label{eq:Tr_Pi_rho_0-rho_1-2}\\
&\ge \|\rho_0-\rho_1\|_{tr}-8\delta \label{eq:Tr_Pi_rho_0-rho_1-3}\\
&=\frac{1}{2p(\secp)}  \label{eq:Tr_Pi_rho_0-rho_1-4}
\end{align}
where \Cref{eq:Tr_Pi_rho_0-rho_1-2} follows from \Cref{eq:Tr_Pi_M_0-M_1-1,eq:Tr_Pi_M_0-M_1-2,eq:Tr_Pi_M_0-M_1-3,eq:Tr_Pi_M_0-M_1-4,eq:Tr_Pi_M_0-M_1-5} and \Cref{cor:trace_of_projection_and_trace_distance},
\Cref{eq:Tr_Pi_rho_0-rho_1-3} follows from the assumption that  $\|M_b-\rho_b\|_{tr}\le \delta$ for $b\in \bit$, 
and \Cref{eq:Tr_Pi_rho_0-rho_1-4} follows from  $\|\rho_0-\rho_1\|_{tr}\ge\frac{1}{p(\secp)}$ and $\delta=\frac{1}{16p(\secp)}$. 
\end{proof}

\subsection{Upper Bound}
We show a matching upper bound assuming the existence of exponentially secure PRGs.  
\begin{theorem}[Restatement of \Cref{thm:intro_upper_EFI}]\label{thm:upper_EFI}
If exponentially secure PRGs exist, then
$m(\secp)$-qubit 
EFI pairs exist for any QPT computable function $m(\secp)=\omega(\log\secp)$.
\end{theorem}
\begin{proof}
Let $G$ be an exponentially secure PRG. Without loss of generality, we assume that $G$ is length-doubling, i.e., it maps $n$-bit strings to $2n$-bit strings.\footnote{This is because we can stretch the output length by sequential application which preserves the exponential security.} 
The exponential security means that
there exists a constant $0<c<1$ such that for any QPT adversary $\cA$
\begin{align}
\left|\Pr_{x\gets\bit^\secp}[1\gets\cA(1^\secp,G(x))]    
-\Pr_{y\gets\bit^{2\secp}}[1\gets\cA(1^\secp,y)]    \right|\le 2^{-c\secp}
\end{align}
for all sufficiently large $\secp$.
Let $n(\secp):=\lceil m(\secp)/2 \rceil$.  
Then we have
\begin{align}
\left|\Pr_{x\gets\bit^n}[1\gets\cA(1^n,G(x))]    
-\Pr_{y\gets\bit^{2n}}[1\gets\cA(1^n,y)]    \right|\le 2^{-cn}=\negl(\secp) \label{eq:PRG_security}
\end{align}
for sufficiently large $\secp$. 
Define EFI pairs as follows:
\begin{align}
\rho_0&\coloneqq\frac{1}{2^n} \sum_{x\in\bit^n}\ket{G(x)}\bra{G(x)},\\
\rho_0&\coloneqq\frac{1}{2^{2n}} \sum_{y\in\bit^{2n}}\ket{y}\bra{y}
\end{align}
The computational indistinguishability immediately follows from \Cref{eq:PRG_security}.
On the other hand,  \cref{lem:fidelity_upperbound} implies
\begin{align}
F(\rho_0,\rho_1)\le 2^{-n}    =\negl(\secp)
\end{align}
and thus $\|\rho_0-\rho_1\|_{tr}=1-\negl(\secp)$.  
Thus, they are stastically far. 
Finally, $\rho_0$ and $\rho_1$ are $2n$ qubit states. 
When $m$ is even, then $2n=m$ and thus it gives $m$-qubit EFI pairs. 
When $m$ is odd, then $2n=m-1$, but we can add append one-qubit state $\ket{0}$ to both $\rho_0$ and $\rho_1$ so that they are $m$-qubit states while preserving computational indistinguishability and statistical farness. 
\end{proof}


\ifnum\anonymous=1
\else
{\bf Acknowledgements.}
This work was initiated while the authors were joining NII Shonan Meeting No.198 ``New Directions in Provable Quantum Advantages''.
We thank Ryu Hayakawa and Fermi Ma for useful conversations. 
We thank anonymous reviewers of CRYPTO 2024 and TQC 2024 for their valuable comments, especially for pointing out that injectivity is not needed in \Cref{thm:upperbound_weak}.
MH is supported by a KIAS Individual Grant QP089802.
TM is supported by
JST CREST JPMJCR23I3,
JST Moonshot R\verb|&|D JPMJMS2061-5-1-1, 
JST FOREST, 
MEXT QLEAP, 
the Grant-in Aid for Transformative Research Areas (A) 21H05183,
and 
the Grant-in-Aid for Scientific Research (A) No.22H00522.
\fi

\ifnum\submission=0
\bibliographystyle{alpha} 
\else
\bibliographystyle{splncs04}
\fi
\bibliography{abbrev3,crypto,reference}

\newcommand{\etalchar}[1]{$^{#1}$}
\begin{thebibliography}{BCWdW01}

\bibitem[AGQY22]{TCC:Luowen}
Prabhanjan Ananth, Aditya Gulati, Luowen Qian, and Henry Yuen.
\newblock Pseudorandom (function-like) quantum state generators: New definitions and applications.
\newblock TCC, 2022.

\bibitem[AK07]{AK07}
Scott Aaronson and Greg Kuperberg.
\newblock Quantum versus classical proofs and advice.
\newblock {\em Theory OF Computing}, 3:129--157, 2007.

\bibitem[ALY24]{cryptoeprint:2023/904}
Prabhanjan Ananth, Yao-Ting Lin, and Henry Yuen.
\newblock Pseudorandom strings from pseudorandom quantum states.
\newblock In {\em 15th Innovations in Theoretical Computer Science Conference (ITCS 2024)}. Schloss-Dagstuhl-Leibniz Zentrum f{\"u}r Informatik, 2024.

\bibitem[AQY22]{C:AnaQiaYue22}
Prabhanjan Ananth, Luowen Qian, and Henry Yuen.
\newblock Cryptography from pseudorandom quantum states.
\newblock In Yevgeniy Dodis and Thomas Shrimpton, editors, {\em CRYPTO~2022, Part~I}, volume 13507 of {\em {LNCS}}, pages 208--236. Springer, Heidelberg, August 2022.

\bibitem[AS17]{ABMB}
Guillaume Aubrun and Stanis{\l}aw~J Szarek.
\newblock {\em Alice and Bob meet Banach}, volume 223.
\newblock American Mathematical Soc., 2017.

\bibitem[BCQ23]{ITCS:BCQ23}
Zvika Brakerski, Ran Canetti, and Luowen Qian.
\newblock On the computational hardness needed for quantum cryptography.
\newblock ITCS 2023, 2023.

\bibitem[BCWdW01]{fingerprinting}
Harry Buhrman, Richard Cleve, John Watrous, and Ronald de~Wolf.
\newblock Quantum fingerprinting.
\newblock {\em Phys. Rev. Lett.}, 87:167902, 2001.

\bibitem[BDF{\etalchar{+}}11]{AC:BDFLSZ11}
Dan Boneh, {\"O}zg{\"u}r Dagdelen, Marc Fischlin, Anja Lehmann, Christian Schaffner, and Mark Zhandry.
\newblock Random oracles in a quantum world.
\newblock In Dong~Hoon Lee and Xiaoyun Wang, editors, {\em ASIACRYPT~2011}, volume 7073 of {\em {LNCS}}, pages 41--69. Springer, Heidelberg, December 2011.

\bibitem[BS08]{BelovsS08}
Aleksandrs Belovs and Juris Smotrovs.
\newblock A criterion for attaining the welch bounds with applications for mutually unbiased bases.
\newblock In Jacques Calmet, Willi Geiselmann, and J{\"{o}}rn M{\"{u}}ller{-}Quade, editors, {\em Mathematical Methods in Computer Science, {MMICS} 2008, Karlsruhe, Germany, December 17-19, 2008 - Essays in Memory of Thomas Beth}, volume 5393 of {\em Lecture Notes in Computer Science}, pages 50--69. Springer, 2008.

\bibitem[BS20]{C:BraShm20}
Zvika Brakerski and Omri Shmueli.
\newblock Scalable pseudorandom quantum states.
\newblock In Daniele Micciancio and Thomas Ristenpart, editors, {\em CRYPTO~2020, Part~II}, volume 12171 of {\em {LNCS}}, pages 417--440. Springer, Heidelberg, August 2020.

\bibitem[CGG{\etalchar{+}}23]{CGGHLP23}
Bruno Cavalar, Eli Goldin, Matthew Gray, Peter Hall, Yanyi Liu, and Angelos Pelecanos.
\newblock On the computational hardness of quantum one-wayness.
\newblock {\em arXiv:2312.08363}, 2023.

\bibitem[DMS00]{EC:DumMaySal00}
Paul Dumais, Dominic Mayers, and Louis Salvail.
\newblock Perfectly concealing quantum bit commitment from any quantum one-way permutation.
\newblock In Bart Preneel, editor, {\em EUROCRYPT~2000}, volume 1807 of {\em {LNCS}}, pages 300--315. Springer, Heidelberg, May 2000.

\bibitem[Gol90]{Gol90}
Oded Goldreich.
\newblock A note on computational indistinguishability.
\newblock Information Processing Letters 34.6 (1990), pp.277–281., 1990.

\bibitem[Gol01]{DBLP:books/cu/Goldreich2001}
Oded Goldreich.
\newblock {\em The Foundations of Cryptography - Volume 1: Basic Techniques}.
\newblock Cambridge University Press, 2001.

\bibitem[HMY23]{EC:HhaMorYam23}
Minki Hhan, Tomoyuki Morimae, and Takashi Yamakawa.
\newblock From the hardness of detecting superpositions to cryptography: Quantum public key encryption and commitments.
\newblock In Carmit Hazay and Martijn Stam, editors, {\em EUROCRYPT~2023, Part~I}, volume 14004 of {\em {LNCS}}, pages 639--667. Springer, Heidelberg, April 2023.

\bibitem[IL89]{FOCS:ImpLub89}
Russell Impagliazzo and Michael Luby.
\newblock One-way functions are essential for complexity based cryptography (extended abstract).
\newblock In {\em 30th FOCS}, pages 230--235. {IEEE} Computer Society Press, October~/~November 1989.

\bibitem[ILL89]{STOC:ImpLevLub89}
Russell Impagliazzo, Leonid~A. Levin, and Michael Luby.
\newblock Pseudo-random generation from one-way functions (extended abstracts).
\newblock In {\em 21st ACM STOC}, pages 12--24. {ACM} Press, May 1989.

\bibitem[JLS18]{C:JiLiuSon18}
Zhengfeng Ji, Yi-Kai Liu, and Fang Song.
\newblock Pseudorandom quantum states.
\newblock In Hovav Shacham and Alexandra Boldyreva, editors, {\em CRYPTO~2018, Part~III}, volume 10993 of {\em {LNCS}}, pages 126--152. Springer, Heidelberg, August 2018.

\bibitem[KQST23]{STOC23:KreQiaSinTal}
William Kretschmer, Luowen Qian, Makrand Sinha, and Avishay Tal.
\newblock Quantum cryptography in algorithmica.
\newblock STOC, 2023.

\bibitem[Kre21]{Kre21}
W.~Kretschmer.
\newblock Quantum pseudorandomness and classical complexity.
\newblock {\em TQC 2021}, 2021.

\bibitem[LMW23]{cryptoeprint:2023/1602}
Alex Lombardi, Fermi Ma, and John Wright.
\newblock A one-query lower bound for unitary synthesis and breaking quantum cryptography.
\newblock Cryptology ePrint Archive, Paper 2023/1602, 2023.
\newblock \url{https://eprint.iacr.org/2023/1602}.

\bibitem[Low21]{Lowe21}
Angus Lowe.
\newblock Learning quantum states without entangled measurements.
\newblock Master's thesis, University of Waterloo, 2021.

\bibitem[LR86]{STOC:LubRac86}
Michael Luby and Charles Rackoff.
\newblock Pseudo-random permutation generators and cryptographic composition.
\newblock In {\em 18th ACM STOC}, pages 356--363. {ACM} Press, May 1986.

\bibitem[MY22]{C:MorYam22}
Tomoyuki Morimae and Takashi Yamakawa.
\newblock Quantum commitments and signatures without one-way functions.
\newblock In Yevgeniy Dodis and Thomas Shrimpton, editors, {\em CRYPTO~2022, Part~I}, volume 13507 of {\em {LNCS}}, pages 269--295. Springer, Heidelberg, August 2022.

\bibitem[MY24]{TQC:MorYam24}
Tomoyuki Morimae and Takashi Yamakawa.
\newblock One-wayness in quantum cryptography.
\newblock In {\em TQC 2024}, 2024.
\newblock (to appear).

\bibitem[NC10]{Nielsen_Chuang}
Michael Nielsen and Isaac Chuang.
\newblock {\em Quantum Computation and Quantum Information}.
\newblock Cambridge University Press, 2010.

\bibitem[Ras07]{Rastegin_2007}
Alexey~E Rastegin.
\newblock Trace distance from the viewpoint of quantum operation techniques.
\newblock {\em Journal of Physics A: Mathematical and Theoretical}, 40(31):9533, 2007.

\bibitem[Wel74]{Welch}
L.~Welch.
\newblock Lower bounds on the maximum cross correlation of signals (corresp.).
\newblock {\em IEEE Transactions on Information Theory}, 20(3):397--399, 1974.

\bibitem[Yan22]{AC:Yan22}
Jun Yan.
\newblock General properties of quantum bit commitments (extended abstract).
\newblock In Shweta Agrawal and Dongdai Lin, editors, {\em ASIACRYPT~2022, Part~IV}, volume 13794 of {\em {LNCS}}, pages 628--657. Springer, Heidelberg, December 2022.

\bibitem[Yao82]{FOCS:Yao82a}
Andrew Chi-Chih Yao.
\newblock Theory and applications of trapdoor functions (extended abstract).
\newblock In {\em 23rd FOCS}, pages 80--91. {IEEE} Computer Society Press, November 1982.

\end{thebibliography}

\appendix
\section{Quantum Bit Commitments}\label{sec:commitment}
\subsection{Preliminaries}
\begin{definition}[Canonical Quantum Bit Commitments~\cite{AC:Yan22}]
A canonical quantum bit commitment scheme is   
a family $\{Q_0(\secp),Q_1(\secp)\}_{\secp\in{\mathbb N}}$ of efficiently implementable unitary operators that is used as follows.
\begin{itemize}
    \item[{\bf Commit phase:}]
    In the commit phase, the sender who wants to commit $b\in\bit$ generates the state
    $\ket{\Psi_b}_{\regR,\regC}\coloneqq Q_b(\secp)|0...0\rangle$ over two registers $\regR$ and $\regC$.
    The sender sends the commitment register $\regC$ to the receiver.
    \item[{\bf Reveal phase:}]
    In the reveal phase, the sender sends the reveal register $\regR$ and $b$ to the receiver.
    The receiver applies $Q_b^\dagger(\secp)$ on the state over $\regR$ and $\regC$,
    and measures all qubits in the computational basis. If all results are 0,
    the receiver accepts $b$.
    Otherwise, the receiver rejects.
\end{itemize}
\end{definition}
For the notational simplicity, we often omit the security parameter $\secp$ of $\{Q_0(\secp),Q_1(\secp)\}_{\secp\in{\mathbb N}}$,
and just write as $(Q_0,Q_1)$.

\begin{definition}[Hiding]
We say that a commitment scheme $(Q_0,Q_1)$ is statistically (resp. computationally) hiding if    
\begin{align}
|\Pr[1\gets\cA(1^\secp,\rho_0)]    
-\Pr[1\gets\cA(1^\secp,\rho_1)]|\le\negl(\secp)    
\end{align}
for any unbounded (resp. QPT) adversary $\cA$, where
$\rho_b\coloneqq \Tr_{\regR}(\ket{\Psi_b}_{\regR,\regC})$,
and $\Tr_{\regA}(\sigma_{\regA,\regB})$
means tracing out the register $\regA$ of the state $\sigma_{\regA,\regB}$ over two registers $\regA$ and $\regB$.
\end{definition}

\begin{definition}[Binding]
\label{def:honest_binding}
We say that a commitment scheme $(Q_0,Q_1)$ is statistically (resp. computationally) binding if    
the winning probability of the following security game is $\negl(\secp)$ for any unbounded (resp. QPT) adversary $\cA$:  
\begin{enumerate}
    \item 
    $\cA$ (who plays the role of the sender) honestly commits 0.
    \item 
    $\cA$ applies any operation on the register $\regR$.
    \item 
    $\cA$ opens the commitment to 1. 
    \item 
    $\cA$ wins
    if and only if the receiver accepts 1.
\end{enumerate}

\end{definition}
\begin{remark}
The above definition of binding, so-called honest binding, does not capture all possible attacks,
because we assume that the adversary honestly commits 0 in the commit phase.
However, when we want to show lower bounds, restricting to the honest binding only strengthens our results, because any reasonable definition of binding should satisfy the honest binding.
Moreover, it is shown in \cite{AC:Yan22} that the honest binding implies a standard definition of binding, so-called sum-binding~\cite{EC:DumMaySal00}, which captures general attacks. 
In the following, if we just say binding, it means honest binding.
\end{remark}
\begin{remark} \takashi{I added this remark}
In the original definition of canonical bit commitments in \cite{AC:Yan22}, computational hiding and computational binding are required to hold against non-uniform QPT adversaries whereas we only require them against uniform QPT adversaries. 
This only makes our result on the lower bound for the commitment length stronger.  
\end{remark}

\begin{lemma}[Flavor Conversion~\cite{EC:HhaMorYam23}]
\label{lem:flavor}
Let $(Q_0,Q_1)$ be a statistically hiding (resp. binding) and computationally binding (resp. hiding) non-interactive quantum bit commitment scheme.
Then, $(\tilde{Q}_0,\tilde{Q}_1)$ defined as follows is
a statistically-binding (resp. hiding) and computationally-hiding (resp. binding) non-interactive quantum bit commitment scheme.
\begin{align}
\ket{\tilde{\Psi}_b}_{\regR',\regC'}\coloneqq 
\tilde{Q}_b|0...0\rangle\coloneqq
\frac{1}{\sqrt{2}}\left[
|0\rangle_\regD \otimes(Q_0|0...0\rangle)_{\regR,\regC} 
+ (-1)^b|1\rangle_\regD \otimes(Q_1|0...0\rangle)_{\regR,\regC}   
\right]
\end{align}
for each $b\in\bit$,
where $\regR'\coloneqq \regC$
and $\regC'\coloneqq (\regR,\regD)$.
\end{lemma}

\begin{lemma}\label{lem:TrF}
For any two states $\rho$ and $\sigma$, $\Tr(\rho\sigma)\le F(\rho,\sigma)$,    
where $F$ is the fidelity.
\end{lemma}

\begin{proof}
The operator $\sqrt{\sigma}\rho\sqrt{\sigma}$ is positive semi-definite.    
Therefore, it has the diagonalization $\sqrt{\sigma}\rho\sqrt{\sigma}=\sum_i \alpha_i |\phi_i\rangle\langle \phi_i|$,
where $\alpha_i\ge0$ for all $i$.
Hence
\begin{align}
F(\rho,\sigma)=\left(\Tr\sqrt{\sqrt{\sigma}\rho\sqrt{\sigma}}\right)^2    
=\left(\sum_i\sqrt{\alpha_i}\right)^2    
=\sum_i\alpha_i+\sum_{i\neq j}\sqrt{\alpha_i\alpha_j}    
\ge\sum_i\alpha_i
=\Tr(\rho\sigma).
\end{align}
\end{proof}

\subsection{Lower Bound}
\begin{theorem}\label{thm:lower_commitment_C}
There do not exist computationally-binding and computationally-hiding canonical quantum bit commitments 
with $O(\log\secp)$-qubit commitment register.
\end{theorem}

By the flavor conversion (\Cref{lem:flavor}), the above theorem implies the following corollary.
\begin{corollary}\label{cor:lower_commitment_R}
There do not exist computationally-binding and computationally-hiding canonical quantum bit commitments 
with $O(\log\secp)$-qubit reveal register.
\end{corollary}

Below, we prove  \Cref{thm:lower_commitment_C}.

\begin{proof}[Proof of \Cref{thm:lower_commitment_C}]
Assume that    
such a commitment scheme $\{Q_0(\secp),Q_1(\secp)\}_{\secp\in\mathbb{N}}$ exists,
and let $\ket{\Psi_{\secp,b}}_{\regR,\regC}\coloneqq Q_b(\secp)|0...0\rangle$.
Let $\rho^b\coloneqq \Tr_{\regR}(\ket{\Psi_{\secp,b}}_{\regR,\regC})$.\footnote{$\rho_b$ also depends on $\secp$, but we omit the dependence from the notation for simplicity.}
We consider two cases separately.
\begin{enumerate}
\item 
There exists a polynomial $p$ such that
$\left\|\rho_0-\rho_1\right\|_{tr}\ge\frac{1}{p(\secp)}$ for infinitely many $\secp$.
    \item 
$\left\|\rho_0-\rho_1\right\|_{tr}\le\negl(\secp)$.
\end{enumerate}
In the first case, we can show that the scheme does not satisfy computational hiding by exactly the same proof as that of \Cref{thm:lower_EFI}. The only difference is that our assumption is $\left\|\rho_0-\rho_1\right\|_{tr}\ge\frac{1}{p(\secp)}$ for infinitely many $\secp$ instead of for all $\secp$. Thus, by repeating exactly the same argument as in the proof of \Cref{thm:lower_EFI}, we can obtain a QPT adversary that distinguishes $\rho_0$ and $\rho_1$ with advantage at least $\frac{1}{2p(\secp)}$ for infinitely many $\secp$. This means that the scheme does not satisfy computational hiding.

In the second case, the scheme satisfies statistical hiding. 
We prove that it cannot satisfy computational binding. Toward contradiction, we assume that it satisfies computational binding. 
Then, due to the flavor conversion (\cref{lem:flavor}), we have
a statistically-binding and computationally-hiding canonical quantum bit commitment scheme
$\{\tilde{Q}_0(\secp),\tilde{Q}_1(\secp)\}_{\secp\in\mathbb{N}}$ with $|\regR'|=|\regC|=O(\log\secp)$,
where $|\regA|$ is the number of qubits of the register $\regA$.
However, we can show that such commitments do not exist, because we can construct a QPT adversary that
breaks the computational hiding as follows:
\begin{enumerate}
    \item 
    Given $\tilde{\rho}_b\coloneqq\Tr_{\regR'}(\ket{\tilde{\Psi}_b}_{\regR',\regC'})$ as input.
    \item 
    Run the swap test between $\tilde{\rho}_0$ and $\tilde{\rho}_b$.
    \item 
    Output the output of the swap test.
\end{enumerate}
Then the probability $p_b$ that the adversary outputs 1 is
   $p_0=\frac{1+\Tr(\tilde{\rho}_0^2)}{2}$ and
   $p_1=\frac{1+\Tr(\tilde{\rho}_0\tilde{\rho}_1)}{2}$. 
Because of the statistical binding, we have $F(\tilde{\rho}_0,\tilde{\rho}_1)\le\negl(\secp)$. Therefore, by using \cref{lem:TrF},
\begin{align}
\Tr(\tilde{\rho}_0\tilde{\rho}_1)\le F(\tilde{\rho}_0,\tilde{\rho}_1)\le\negl(\secp).    
\end{align}
On the other hand, 
let $\tilde{\rho}_0=\sum_{i\in [N]} p_i \ket{v_i}\bra{v_i}$ be the spectral decomposition of $\tilde{\rho}_0$ where $N\le 2^{|\regR'|}=\poly(\secp)$. 
Then 
\begin{align}
\Tr(\tilde{\rho}_0^2)=
\sum_{i\in [N]} p_i^2
\ge \frac{\left(\sum_{i\in [N]} p_i\right)^2}{N}=\frac{1}{\poly(\secp)}   
\end{align}
where we used the Cauchy-Schwarz inequality. 
Therefore the QPT adversary can break the computational hiding. 
Thus,  the commitment scheme $\{Q_0(\secp),Q_1(\secp)\}_{\secp\in\mathbb{N}}$ does not satisfy computational binding.

In summary, $\{Q_0(\secp),Q_1(\secp)\}_{\secp\in\mathbb{N}}$ does not satisfy computational hiding in the first case and does not satisfy computational binding in the second case. This completes the proof of \Cref{thm:lower_commitment_C}.
\end{proof}
\begin{remark}
It is known that any (even interactive) quantum commitment schemes can be converted into canonical one~\cite{AC:Yan22}, but
the transformation could increase the commitment size. (Honest receiver's register in the original commitment scheme
is contained in the commitment register of the resulting one.) Therefore 
\Cref{thm:lower_commitment_C} 
does not directly exclude 
interactive quantum commitments with  $O(\log\secp)$-qubit communication.
It is an open problem whether interactive quantum commitments with
$O(\log\secp)$-qubit communication are possible or not.
\end{remark}

\subsection{Upper Bound}
Similarly to the case of EFI pairs, 
we show a matching upper bound assuming the existence of exponentially secure PRGs.  
\begin{theorem}\label{thm:upper_bound_com}
If exponentially secure PRGs exist, then
statistically (resp. computationally) binding and computationally (resp. statistically) hiding canonical quantum bit commitments 
with $m(\secp)$-qubit commitment register exist for any QPT computable function $m(\secp)=\omega(\log\secp)$.
\end{theorem}
\begin{proof}
The proof for the case of statistically binding and computationally hiding commitments is similar to that of \Cref{thm:upper_EFI}. 
Let $G$ be an exponentially secure PRG. Without loss of generality, we assume that $G$ is length-doubling. 
   Let $n(\secp):=\lceil m(\secp)/2 \rceil$ and 
   define the canonical quantum bit commitment $(Q_0,Q_1)$ as follows:
\begin{align}
\ket{\Psi_0}_{\regR,\regC}&\coloneqq\frac{1}{\sqrt{2^n}} \sum_{x\in\bit^n}|x0^n\rangle_\regR\otimes |G(x)\rangle_\regC,\label{com1}\\
\ket{\Psi_1}_{\regR,\regC}&\coloneqq\frac{1}{\sqrt{2^{2n}}} \sum_{y\in\bit^{2n}}|y\rangle_\regR\otimes|y\rangle_\regC.\label{com2}
\end{align}
Its computational hiding follows from exponential security of $G$ and $n=\omega(\log \secp)$. 
For its statistical binding, from \cref{lem:fidelity_upperbound}, we have
\begin{align}
F(\rho_0,\rho_1)\le 2^{-n}    =\negl(\secp),
\end{align}
which implies statistical binding (e.g., see \cite{C:MorYam22}).
The number of qubits of the commitment register is $2n$. If $m$ is even, $2n=m$ and if $m$ is odd, $2n=m-1$, so we can simply append one-qubit to the commitment register to make it $m$-qubit. 

The case of computationally binding and statistically hiding commitments immediately from the above and flavor conversion (\Cref{lem:flavor}) noting that $\regC$ and $\regR$ have the same length in the above construction. 
\end{proof}

\clearpage
\newpage
\setcounter{tocdepth}{2}
\tableofcontents

\end{document}